\documentclass[sigconf,natbib=true,anonymous=false]{acmart}

\usepackage{booktabs} 
\usepackage{enumitem}
\usepackage[font=small,labelfont=bf]{caption}
\usepackage{subcaption}
\usepackage{dirtytalk}

\usepackage[utf8]{inputenc}
\usepackage{booktabs}
\usepackage{multirow}
\usepackage{graphics}
\usepackage{natbib}
\usepackage{blindtext}
\usepackage{bm}
\usepackage{mathtools}
\usepackage{algpseudocode}
\usepackage{algorithm}
\usepackage{enumitem} 
\setlist[1]{leftmargin=*} 
\usepackage[symbols,acronym,section]{glossaries}
\makeglossaries
\usepackage{pgfplots}
\usepackage{tikz}
\usetikzlibrary{arrows,calc,rightangles}
\tikzset{
>=stealth',
help lines/.style={dashed, thick},
axis/.style={<->},
important line/.style={thick},
connection/.style={thick, dotted},
}

\usepackage[conf={one big link}]{proofs-at-the-end}  
\pgfkeys{/prAtEnd/global custom defaults/.style={one big link={See proof in Appendix (page~\pageref*{proof:prAtEnd\pratendcountercurrent}).}
}
}
\usepackage[bold]{vectors}
\usepackage[capitalise,nameinlink,noabbrev]{cleveref}
\newcommand{\expo}{\texttt{EXPO}}
\newcommand{\ctrl}{\texttt{CTRL}}
\newcommand{\pl}{\texttt{PL}}
\newcommand{\ltr}{\texttt{LTR}}

\newcommand{\unfairness}{\mathrm F}
\newcommand{\utility}{\mathrm U}
\newcommand{\nutility}{\mathrm{nU}}
\newcommand{\nunfairness}{\mathrm{nF}}
\newcommand{\temperature}{\tau}
\newcommand{\horizon}{T}
\newcommand{\docix}{i}  
\newcommand{\rankix}{k}  
\newcommand{\splits}{S}
\newcommand{\zone}{Z}  

\newcommand{\argsort}{\texttt{argsort}}

\newcommand{\defkey}[1]{\emph{#1}}
\newcommand{\seq}[2]{\{#1,\hdots,#2\}}

\usepackage{mathtools}
\DeclarePairedDelimiter{\norm}{\lVert}{\rVert}

\usepackage{physics} 

\newcommand{\piprp}{\pi_{\PRP}}

\newcommand{\ndoc}{n}

\DeclareMathAlphabet\mathbfcal{OMS}{cmsy}{b}{n}

\renewcommand{\vrho}{\bm{\rho}}
\renewcommand{\vone}{\mathbf{1}}

\renewcommand{\vcE}{\mathbfcal{E}}
\newcommand{\vcEt}{\mathbfcal{E}^*}

\DeclareMathOperator{\PL}{PL}

\DeclareMathOperator{\PRP}{PRP}

\DeclareMathOperator{\Conv}{Conv}

\theoremstyle{definition}
\theoremstyle{theorem}\newtheorem{lemm}{Lemma}
\theoremstyle{theorem}
\theoremstyle{theorem}
\theoremstyle{theorem}
\theoremstyle{remark}\newtheorem{rema}{Remark}

\newcommand{\defeq}{\vcentcolon=}

\newcommand{\splitatcommas}[1]{%
  \begingroup
  \begingroup\lccode`~=`, \lowercase{\endgroup
    \edef~{\mathchar\the\mathcode`, \penalty0 \noexpand\hspace{0pt plus 1em}}%
  }\mathcode`,="8000 #1%
  \endgroup
}

\makeatletter

\newcommand\thefontsize[1]{{#1 The current font size is: \f@size pt\par}}

\makeatother

\newacronym{acm}{ACM}{Association for Computing Machinery}
\newacronym{adaboost}{AdaBoost}{Adaptative Boosting}
\newacronym{ads}{ADS}{Automated Decision System}
\newacronym{ai}{AI}{Artificial Intelligence}
\newacronym{aif}{AIF}{Average Individual Fairness}
\newacronym{aka}{a.k.a.}{also known as}
\newacronym{ama}{AMA}{Artificial Moral Agent}
\newacronym{ap}{AP}{Average Precision}
\newacronym{bert}{BERT}{Bidirectional Encoder Representations from Transformers}
\newacronym{bisg}{BISG}{Bayesian Improved Surname Geocoding}
\newacronym{bm}{BM}{Best Matching}
\newacronym{bo}{BO}{Bayesian Optimisation}
\newacronym{bpr}{BPR}{Bayesian Personalized Ranking}
\newacronym{bp}{BP}{Belief Propagation}
\newacronym{bvn}{BvN}{Birkhoff-von Neumann}
\newacronym{cb}{CB}{Contextual Bandits}
\newacronym{cca}{CCA}{Canonical Correlation Analysis}
\newacronym{ccm}{CCM}{Click Chain Model}
\newacronym{cdssm}{C-DSSM}{Convolutional Deep Structured Semantic Model}
\newacronym{cf}{CF}{Collaborative Filtering}
\newacronym{cf.}{cf.}{confer}
\newacronym{cikm}{CIKM}{International Conference on Information and Knowledge Management}
\newacronym{clara}{CLARA}{CLick And RAnking}
\newacronym{clef}{CLEF}{Conference and Labs of the Evaluation Forum}
\newacronym{clm}{CLM}{Cumulative Link Model}
\newacronym{clsm}{CLSM}{Convolutional Latent Semantic Model}
\newacronym{cm}{CM}{Cascade Model}
\newacronym{cnn}{CNN}{Convolutional Neural Network}
\newacronym{compas}{COMPAS}{Correctional Offender Management Profiling for Alternative Sanctions}
\newacronym{cpu}{CPU}{Central Processing Unit}
\newacronym{crs}{CRS}{Conversational Recommender System}
\newacronym{csc}{CSC}{Cost Sensitive Classification}
\newacronym{csi}{CSI}{Comité de suivi de thèse}
\newacronym{css}{CSS}{Complementarity Sum Sampling}
\newacronym{ctr}{CTR}{Click Through Rate}
\newacronym{ctlr}{Ctrl}{Controller}
\newacronym{da}{DA}{Data Augmentation}
\newacronym{dag}{DAG}{Directed Acyclic Graph}
\newacronym{dbn}{DBN}{Dynamic Bayesian Network}
\newacronym{dcg}{DCG}{Discounted Cumulative Gain}
\newacronym{dcm}{DCM}{Discrete Choice Model}
\newacronym{dct}{DCT}{Discrete Cosine Transform}
\newacronym{dlrm}{DLRM}{Deep Learning Recommendation Model}
\newacronym{dm}{DM}{Decision Maker}
\newacronym{dnn}{DNN}{Deep Neural Network}
\newacronym{dp}{DP}{Differential Privacy}
\newacronym{dssm}{DSSM}{Deep Structured Semantic Model}
\newacronym{ebm}{EBM}{Energy-Based Model}
\newacronym{ebms}{EBMs}{Energy-Based Models}
\newacronym{ece}{ECE}{Empirical Calibration Error}
\newacronym{ecir}{ECIR}{European Conference on Information Retrieval}
\newacronym{ef}{EF}{Envy-Free}
\newacronym{ef1}{EF1}{Envy-Free up to one item}
\newacronym{efr}{EFR}{Expected First Relevant}
\newacronym{eg}{e.g.}{exempli gratia}
\newacronym{elmo}{ELMo}{Embeddings from Language Models}
\newacronym{erm}{ERM}{Empirical Risk Minimization}
\newacronym{err}{ERR}{Expected Reciprocal Rank}
\newacronym{ers}{ERS}{Empirical Relational System}
\newacronym{etc}{etc.}{et cetera}
\newacronym{fair}{FAIR}{Fairness-Aware Information Retrieval}
\newacronym{fia}{FIA}{Fast Influence Analysis}
\newacronym{fire}{FIRE}{Forum for Information Retrieval Evaluation}
\newacronym{fm}{FM}{Factorization Machine}
\newacronym{fultr}{FULTR}{Fair Un-biased Learning-to-Rank}
\newacronym{gan}{GAN}{Generative Adversarial Network}
\newacronym{gbdt}{GBDT}{Gradient Boosted Decision Tree}
\newacronym{gcn}{GCN}{Graph Convolution Network}
\newacronym{gdpr}{GDPR}{General Data Protection Regulation}
\newacronym{gls}{GLS}{Grötschel, Lovász and Schrijver}
\newacronym{gnn}{GNN}{Graph Neural Network}
\newacronym{gpt}{GPT}{Generative Pre-trained Transformer}
\newacronym{gpu}{GPU}{Graphics Processing Unit}
\newacronym{hpc}{HPC}{High Performance Computing}
\newacronym{hr}{HR}{Hit Rate}
\newacronym{idcg}{IDCG}{Ideal Discounted Cumulative Gain}
\newacronym{idf}{IDF}{Inverse Document Frequency}
\newacronym{ie}{i.e.,}{id est}
\newacronym{iid}{i.i.d.}{independent and identically distributed}
\newacronym{ilp}{ILP}{Integer Linear Program}
\newacronym{ipm}{IPM}{Integral Probability Metric}
\newacronym{ips}{IPS}{Inverse Propensity Score}
\newacronym{ipu}{IPU}{Intelligence Processing Unit}
\newacronym{ir}{IR}{Information Retrieval}
\newacronym{irgan}{IRGAN}{Information Retrieval Generative Adversarial Network}
\newacronym{irm}{IRM}{Invariant Risk Minimization}
\newacronym{irsvm}{IR-SVM}{Information Retrieval Support Vector Machine}
\newacronym{kl}{KL}{Kullback-Leibler}
\newacronym{knn}{KNN}{$K$-nearest neighbors}
\newacronym{lds}{LDS}{Low Discrepancy Sequence}
\newacronym{lc}{LC}{Latent Cross}
\newacronym{letor}{LETOR}{LEarning TO Rank for Information Retrieval}
\newacronym{lm}{LM}{Language Model}
\newacronym{lp}{LP}{Linear Program}
\newacronym{lsi}{LSI}{Latent Semantic Indexing}
\newacronym{lspi}{LSPI}{Least Squares Policy Iteration}
\newacronym{lstdq}{LSTD-Q}{Least Squares Temporal Difference Q-learning}
\newacronym{lte}{LTE}{Learning To Expose}
\newacronym{ltr}{LTR}{Learning to Rank}
\newacronym{mab}{MAB}{Multi-Armed Bandit}
\newacronym{map}{MAP}{Mean Average Precision}
\newacronym{mcrank}{McRank}{Multi-class Classification for Ranking}
\newacronym{mdp}{MDP}{Markov Decision Process}
\newacronym{me}{ME}{Maximum Entropy}
\newacronym{mf}{MF}{Matrix Factorization}
\newacronym{mfr}{MFR}{Mean First Relevant}
\newacronym{mhr}{MHR}{Multiple Hyperplane Ranker}
\newacronym{mi}{MI}{Mutual Information}
\newacronym{miai}{MIAI}{Multidisciplinary Institute for AI}
\newacronym{milp}{MILP}{Mixed Integer Linear Program}
\newacronym{mimd}{MIMD}{Multiple Instructions Multiple Data}
\newacronym{ml}{ML}{Machine Learning}
\newacronym{mlo}{MLO}{Machine Learning and Optimization}
\newacronym{mlp}{MLP}{MultiLayer Perceptron}
\newacronym{mms}{MMS}{MaxiMin Share}
\newacronym{mnist}{MNIST}{Modified National Institute of Standards and Technology}
\newacronym{molp}{MOLP}{Multi-Objective Linear Program}
\newacronym{moo}{MOO}{Multi-Objective Optimisation}
\newacronym{mpt}{MPT}{Modern Portfolio Theory}
\newacronym{mrr}{MRR}{Mean Reciprocal Rank}
\newacronym{mrs}{MRS}{Mulstistakeholder Recommender System}
\newacronym{mslr}{MSLR}{Microsoft Learning to Rank}
\newacronym{msp}{MSP}{Multi-Sided Platform}
\newacronym{mvd}{MVD}{MultiValued Dependency}
\newacronym{ncf}{NCF}{Neural Collaborative Filtering}
\newacronym{ndcg}{NDCG}{Normalized Discounted Cumulative Gain}
\newacronym{ndd}{nDD}{normalized Discounted Difference}
\newacronym{ndf}{NDF}{Nearest Driver First}
\newacronym{ndkl}{nDKL}{normalized Discounted Kullback-Leibler Divergence}
\newacronym{ndr}{nDR}{normalized Discounted Ratio}
\newacronym{ngcf}{NGCF}{Neural Graph Collaborative Filtering}
\newacronym{nle}{NLE}{Naver Labs Europe}
\newacronym{nlp}{NLP}{Natural Language Processing}
\newacronym{nmf}{NMF}{Non-negative matrix factorization}
\newacronym{nn}{NN}{Neural Network}
\newacronym{np}{NP}{Non-deterministic Polynomial-time }
\newacronym{nrb}{nRB}{normalized Ranking Bias}
\newacronym{nrs}{NRS}{Numerical Relational System}
\newacronym{ntcir}{NTCIR}{NII Testbeds and Community for Information Access Research}
\newacronym{ood}{OOD}{Out-Of-Distribution}
\newacronym{pb}{PB}{Poisson Binomial distribution}
\newacronym{pm}{PM}{Poisson Multinomial distribution}
\newacronym{pbm}{PBM}{Position Based Model}
\newacronym{pdp}{PDP}{Pigou-Dalton Principle}
\newacronym{pf}{PF}{Poisson Factorization}
\newacronym{pl}{PL}{Plackett-Luce}
\newacronym{pmf}{pmf}{probability mass function}
\newacronym{pomdp}{POMDP}{Partially Observed Markov Decision Process}
\newacronym{prairie}{PRAIRIE}{PaRis Artificial Intelligence Research InstitutE}
\newacronym{pranking}{PRanking}{Perceptron based Ranking}
\newacronym{pfs}{PFS}{Proportional Fair Share}
\newacronym{prp}{PRP}{Probability Ranking Principle}
\newacronym{qp}{QP}{Quadratic Programming}
\newacronym{rbm}{RBM}{Restricted Boltzmann Machine}
\newacronym{rbp}{RBP}{Rank Biased Precision}
\newacronym{rcf}{RCF}{Relational Collaborative Filtering}
\newacronym{rcm}{RCM}{Random Click Model}
\newacronym{rct}{RCT}{Randomized Controlled Trial}
\newacronym{rc}{RC}{Rank Correlation}
\newacronym{relu}{ReLU}{Rectified Linear Unit}
\newacronym{resp}{resp.}{respectively}
\newacronym{rl}{RL}{Reinforcment Learning}
\newacronym{rnn}{RNN}{Recurrent Neural Network}
\newacronym{roc}{ROC}{Receiver Operating Characteristic}
\newacronym{rr}{RR}{Round-Robin}
\newacronym{rs}{RS}{Recommender System}
\newacronym{rv}{r.v.}{random variable}
\newacronym{sa}{SA}{Self Attention}
\newacronym{sar}{SaR}{Search and Recommendation}
\newacronym{scm}{SCM}{Structural Causal Model}
\newacronym{sdbn}{SDBN}{Simplified Dynamic Bayesian Network}
\newacronym{sea}{SEA}{Safe Exploration Algorithm}
\newacronym{serp}{SERP}{Search Engine Result Page}
\newacronym{sgd}{SGD}{Stochastic Gradient Descent}
\newacronym{sigir}{SIGIR}{Special Interest Group on Information Retrieval}
\newacronym{sigmod}{SIGMOD}{Special Interest Group on Management of Data}
\newacronym{slim}{SLiM}{Sparse Linear Model}
\newacronym{st}{s.t.}{such that}
\newacronym{svd}{SVD}{Singular Value Decomposition}
\newacronym{svm}{SVM}{Support Vector Machine}
\newacronym{svmmap}{SVM$^{map}$}{Support Vector Machine maximum average precision}
\newacronym{tf}{TF}{Term Frequency}
\newacronym{tgn}{TGN}{Temporal Graph Network}
\newacronym{trec}{TREC}{Text REtrieval Conference}
\newacronym{tsp}{TSP}{Traveling Salesman Problem}
\newacronym{ubm}{UBM}{User Browsing Model}
\newacronym{ui}{UI}{User Interface}
\newacronym{url}{URL}{Uniform Resource Locator}
\newacronym{vae}{VAE}{Variational AutoEncoder}
\newacronym{vi}{VI}{Variational Inequality}
\newacronym{vsm}{VSM}{Vector Space Model}
\newacronym{wdf}{WDF}{Worst-off Driver First}
\newacronym{wlog}{w.l.o.g.}{without loss of generality}
\newacronym{wrt}{w.r.t.}{with respect to}
\newacronym{wsdm}{WSDM}{Web Search and Data Mining}

\newglossaryentry{adarank}{name=AdaRank, description={}}
\newglossaryentry{apprank}{name=AppRank, description={}}
\newglossaryentry{bn-decomp}{name=Birkhoff-von Neumann decomposition, description={A (generally non-unique) decomposition of a doubly stochastic matrix into a convex sum of permutation matrices}}
\newglossaryentry{boosting}{name=Boosting, description={Using weighed prediction of a ton of weak learners to build One strong learner, and rule them all.}}
\newglossaryentry{borda-count}{name=Borda Count, description={A vote aggregation method formalized by the chevalier de Borda}}
\newglossaryentry{jm-smooth}{name=Jelinek-Mercer Smoothing, description={Interpolating with the prior (convex combination)}}
\newglossaryentry{katz}{name={Katz measure}, description={A measure to quantify the \emph{centrality} of a node in a graph}}
\newglossaryentry{listmle}{name=ListMLE, description={}}
\newglossaryentry{listnet}{name=ListNet, description={}}
\newglossaryentry{lstm}{name=Long Short-Term Memory (LSTM), description={Memory cell structure in a \acrlong{nn} \cite{hochreiter_long_1997}. See also \url{https://en.wikipedia.org/wiki/Long_short-term_memory}}}
\newglossaryentry{luce}{name=Luce, description={A probability distribution on permutations}}
\newglossaryentry{oracle-efficient}{name=Oracle-efficient algorithm, description={An algorithm that is given access to any standard, fairness-free learning heuristic, from \cite{kearns_average_2019}}}
\newglossaryentry{ord-reg}{name=Ordinal Regression, description={Regression where the target space is characterized only by the relative order of its elements}}
\newglossaryentry{permurank}{name=PermuRank, description={}}
\newglossaryentry{rankgp}{name=RankGP, description={A genetic algorithm for learning rankings}}
\newglossaryentry{rank-agg}{name=Rank Aggregation, description={The problem of finding the ranking with maximum agreement to (possibly incompatible) pairwise orderings. It has been proven to be NP-hard.}}
\newglossaryentry{rank-boost}{name=RankBoost, description={A Boosting method using AdaBoost \cite{freund_efficient_2003}}}
\newglossaryentry{ranknet}{name=RankNet, description={A learning-to-rank algorithm \cite{burges_learning_2005}}}
\newglossaryentry{reciprocal-recommender}{name={Reciprocal Recommender}, description={A recommender platform where two types of users are recommended to each other heterosexually, \acrshort{eg} workerss and employers, men and women, french russian-learners and russian french-learners}}
\newglossaryentry{shilling}{name={Shilling Attack}, description={When an item producer gives himself lots of high ratings}}
\newglossaryentry{softrank}{name=SoftRank, description={}}
\newglossaryentry{vc-dim}{name=Vapnik-Chervonenkis-dimension, description={cardinality of the largest set of points that the algorithm can shatter, from wiki}}
\newglossaryentry{vertical}{name=vertical, description={Within a websearch, the results are presented in blocks of e.g. videos, news articles, websites, pictures. Those blocks are called verticals}}

\usepackage{balance}


\copyrightyear{2022}
\acmYear{2022}
\setcopyright{acmlicensed}\acmConference[SIGIR '22]{Proceedings of the 45th
International ACM SIGIR Conference on Research and Development in
Information Retrieval}{July 11--15, 2022}{Madrid, Spain}
\acmBooktitle{Proceedings of the 45th International ACM SIGIR Conference on
Research and Development in Information Retrieval (SIGIR '22), July 11--15,
2022, Madrid, Spain}
\acmPrice{15.00}
\acmDOI{10.1145/3477495.3532036}
\acmISBN{978-1-4503-8732-3/22/07}

\begin{document}
\settopmatter{printfolios=true}

\title{Pareto-Optimal Fairness-Utility Amortizations in Rankings with a DBN Exposure Model}


\author{Till Kletti}
\email{till.kletti@naverlabs.com}
\orcid{0000-0002-8853-4618}
\affiliation{%
\institution{Naver Labs Europe}
\state{France}
}

\author{Jean-Michel Renders}
\email{jean-michel.renders@naverlabs.com}
\orcid{0000-0002-7516-3707}
\affiliation{%
\institution{Naver Labs Europe}
\state{France}
}

\author{Patrick Loiseau}
\email{patrick.loiseau@inria.fr}
\affiliation{%
\institution{Univ. Grenoble Alpes, Inria, CNRS, Grenoble INP, LIG}
\state{France}
}

\renewcommand{\shortauthors}{Kletti et al.}

\begin{teaserfigure}
    \vspace{-\baselineskip}
    \begin{subfigure}[b]{0.3\textwidth}
        \centering
        \begin{tikzpicture}
	[scale=5.000000,
	back/.style={loosely dotted, thin},
	edge/.style={color=black, thick},
	facet/.style={fill=white,fill opacity=0.300000},
	vertex/.style={inner sep=1pt,circle,draw=red!25!black,fill=red!75!black,thick}]
  \def\costhirty{0.8660256}

  \colorlet{anglecolor}{green!50!black}
  \colorlet{sincolor}{red}
  \colorlet{tancolor}{orange!80!black}
  \colorlet{coscolor}{blue}

  \tikzstyle{axes}=[]
  \tikzstyle{important line}=[very thick]
  \tikzstyle{information text}=[rounded corners,fill=red!10,inner sep=1ex]


	%
	%


	\draw[edge] (1, 0.185) -- (0.463, 1);
	\node[vertex] at (1, 0.185)     {};
	\node[vertex] at (0.463, 1)     {};



\end{tikzpicture}
        \caption{$\ndoc=2$. The DBN-expohedron in $\R^2$ for a relevance vector $\vrho=(0.9, 0.1)$}
    \end{subfigure}
    \hfill
    \begin{subfigure}[b]{0.3\textwidth}
        \centering
        \begin{tikzpicture}%
	[x={(-0.697158cm, -0.389674cm)},
	y={(0.716917cm, -0.378888cm)},
	z={(-0.000040cm, 0.839404cm)},
	scale=3.000000,
	back/.style={loosely dotted, thin},
	edge/.style={color=black, thick},
	facet/.style={fill=white,fill opacity=0.300000},
	vertex/.style={inner sep=1pt,circle,draw=red!25!black,fill=red!75!black,thick}]
%
%

\coordinate (1.00000, 0.46500, 0.15113) at (1.00000, 0.46500, 0.15113);
\coordinate (1.00000, 0.08603, 0.46500) at (1.00000, 0.08603, 0.46500);
\coordinate (0.32500, 1.00000, 0.15113) at (0.32500, 1.00000, 0.15113);
\coordinate (0.06012, 1.00000, 0.32500) at (0.06012, 1.00000, 0.32500);
\coordinate (0.18500, 0.08603, 1.00000) at (0.18500, 0.08603, 1.00000);
\coordinate (0.06012, 0.18500, 1.00000) at (0.06012, 0.18500, 1.00000);
\fill[facet] (0.06012, 0.18500, 1.00000) -- (0.06012, 1.00000, 0.32500) -- (0.32500, 1.00000, 0.15113) -- (1.00000, 0.46500, 0.15113) -- (1.00000, 0.08603, 0.46500) -- (0.18500, 0.08603, 1.00000) -- cycle {};
\draw[edge] (1.00000, 0.46500, 0.15113) -- (1.00000, 0.08603, 0.46500);
\draw[edge] (1.00000, 0.46500, 0.15113) -- (0.32500, 1.00000, 0.15113);
\draw[edge] (1.00000, 0.08603, 0.46500) -- (0.18500, 0.08603, 1.00000);
\draw[edge] (0.32500, 1.00000, 0.15113) -- (0.06012, 1.00000, 0.32500);
\draw[edge] (0.06012, 1.00000, 0.32500) -- (0.06012, 0.18500, 1.00000);
\draw[edge] (0.18500, 0.08603, 1.00000) -- (0.06012, 0.18500, 1.00000);
\node[vertex] at (1.00000, 0.46500, 0.15113)     {};
\node[vertex] at (1.00000, 0.08603, 0.46500)     {};
\node[vertex] at (0.32500, 1.00000, 0.15113)     {};
\node[vertex] at (0.06012, 1.00000, 0.32500)     {};
\node[vertex] at (0.18500, 0.08603, 1.00000)     {};
\node[vertex] at (0.06012, 0.18500, 1.00000)     {};
\end{tikzpicture}
        \caption{$\ndoc=3$. The DBN-expohedron in $\R^3$ for a relevance vector $\vrho=(0.9, 0.5, 0.1)$}
    \end{subfigure}
    \hfill
    \begin{subfigure}[b]{0.3\textwidth}
        \centering
        \begin{tikzpicture}%
	[x={(-0.822913cm, -0.486010cm)},
	y={(0.568168cm, -0.703865cm)},
	z={(-0.000058cm, 0.518043cm)},
	scale=3.000000,
	back/.style={loosely dotted, thin},
	edge/.style={color=black, thick},
	facet/.style={fill=white,fill opacity=0.300000},
	vertex/.style={inner sep=1pt,circle,draw=red!25!black,fill=red!75!black,thick}]
%
%

\coordinate (0.79467, 0.33304, 0.08619) at (0.79467, 0.33304, 0.08619);
\coordinate (0.79467, 0.33304, -0.08429) at (0.79467, 0.33304, -0.08429);
\coordinate (0.79467, -0.08963, 0.31295) at (0.79467, -0.08963, 0.31295);
\coordinate (0.79467, -0.18679, 0.26521) at (0.79467, -0.18679, 0.26521);
\coordinate (0.79467, -0.10859, -0.30131) at (0.79467, -0.10859, -0.30131);
\coordinate (0.79467, -0.18679, -0.25936) at (0.79467, -0.18679, -0.25936);
\coordinate (0.06592, 0.79213, 0.08619) at (0.06592, 0.79213, 0.08619);
\coordinate (0.06592, 0.79213, -0.08429) at (0.06592, 0.79213, -0.08429);
\coordinate (-0.20777, 0.73041, 0.21181) at (-0.20777, 0.73041, 0.21181);
\coordinate (-0.27068, 0.71622, 0.18536) at (-0.27068, 0.71622, 0.18536);
\coordinate (-0.22005, 0.72764, -0.20451) at (-0.22005, 0.72764, -0.20451);
\coordinate (-0.27068, 0.71622, -0.18128) at (-0.27068, 0.71622, -0.18128);
\coordinate (-0.04744, -0.27956, 0.69946) at (-0.04744, -0.27956, 0.69946);
\coordinate (-0.04744, -0.37672, 0.65172) at (-0.04744, -0.37672, 0.65172);
\coordinate (-0.20777, -0.17856, 0.69946) at (-0.20777, -0.17856, 0.69946);
\coordinate (-0.27068, -0.19275, 0.67301) at (-0.27068, -0.19275, 0.67301);
\coordinate (-0.24102, -0.42038, 0.57034) at (-0.24102, -0.42038, 0.57034);
\coordinate (-0.27068, -0.40170, 0.57034) at (-0.27068, -0.40170, 0.57034);
\coordinate (-0.08523, -0.30705, -0.67123) at (-0.08523, -0.30705, -0.67123);
\coordinate (-0.08523, -0.38524, -0.62928) at (-0.08523, -0.38524, -0.62928);
\coordinate (-0.22005, -0.22212, -0.67123) at (-0.22005, -0.22212, -0.67123);
\coordinate (-0.27068, -0.23354, -0.64799) at (-0.27068, -0.23354, -0.64799);
\coordinate (-0.24102, -0.42038, -0.55777) at (-0.24102, -0.42038, -0.55777);
\coordinate (-0.27068, -0.40170, -0.55777) at (-0.27068, -0.40170, -0.55777);
\draw[edge,back] (0.79467, -0.10859, -0.30131) -- (-0.08523, -0.30705, -0.67123);
\draw[edge,back] (0.79467, -0.18679, -0.25936) -- (-0.08523, -0.38524, -0.62928);
\draw[edge,back] (-0.22005, 0.72764, -0.20451) -- (-0.22005, -0.22212, -0.67123);
\draw[edge,back] (-0.27068, 0.71622, -0.18128) -- (-0.27068, -0.23354, -0.64799);
\draw[edge,back] (-0.24102, -0.42038, 0.57034) -- (-0.24102, -0.42038, -0.55777);
\draw[edge,back] (-0.27068, -0.40170, 0.57034) -- (-0.27068, -0.40170, -0.55777);
\draw[edge,back] (-0.08523, -0.30705, -0.67123) -- (-0.08523, -0.38524, -0.62928);
\draw[edge,back] (-0.08523, -0.30705, -0.67123) -- (-0.22005, -0.22212, -0.67123);
\draw[edge,back] (-0.08523, -0.38524, -0.62928) -- (-0.24102, -0.42038, -0.55777);
\draw[edge,back] (-0.22005, -0.22212, -0.67123) -- (-0.27068, -0.23354, -0.64799);
\draw[edge,back] (-0.27068, -0.23354, -0.64799) -- (-0.27068, -0.40170, -0.55777);
\draw[edge,back] (-0.24102, -0.42038, -0.55777) -- (-0.27068, -0.40170, -0.55777);
\node[vertex] at (-0.08523, -0.30705, -0.67123)     {};
\node[vertex] at (-0.22005, -0.22212, -0.67123)     {};
\node[vertex] at (-0.08523, -0.38524, -0.62928)     {};
\node[vertex] at (-0.27068, -0.23354, -0.64799)     {};
\node[vertex] at (-0.27068, -0.40170, -0.55777)     {};
\node[vertex] at (-0.24102, -0.42038, -0.55777)     {};
\fill[facet] (-0.27068, 0.71622, -0.18128) -- (-0.22005, 0.72764, -0.20451) -- (0.06592, 0.79213, -0.08429) -- cycle {};
\fill[facet] (-0.27068, -0.19275, 0.67301) -- (-0.27068, 0.71622, 0.18536) -- (-0.20777, 0.73041, 0.21181) -- (-0.20777, -0.17856, 0.69946) -- cycle {};
\fill[facet] (-0.27068, -0.40170, 0.57034) -- (-0.27068, -0.19275, 0.67301) -- (-0.20777, -0.17856, 0.69946) -- (-0.04744, -0.27956, 0.69946) -- (-0.04744, -0.37672, 0.65172) -- (-0.24102, -0.42038, 0.57034) -- cycle {};
\fill[facet] (-0.27068, 0.71622, -0.18128) -- (-0.27068, 0.71622, 0.18536) -- (-0.20777, 0.73041, 0.21181) -- (0.06592, 0.79213, 0.08619) -- (0.06592, 0.79213, -0.08429) -- (-0.22005, 0.72764, -0.20451) -- cycle {};
\fill[facet] (-0.04744, -0.37672, 0.65172) -- (0.79467, -0.18679, 0.26521) -- (0.79467, -0.08963, 0.31295) -- (-0.04744, -0.27956, 0.69946) -- cycle {};
\fill[facet] (-0.20777, -0.17856, 0.69946) -- (-0.20777, 0.73041, 0.21181) -- (0.06592, 0.79213, 0.08619) -- (0.79467, 0.33304, 0.08619) -- (0.79467, -0.08963, 0.31295) -- (-0.04744, -0.27956, 0.69946) -- cycle {};
\fill[facet] (0.06592, 0.79213, -0.08429) -- (0.79467, 0.33304, -0.08429) -- (0.79467, 0.33304, 0.08619) -- (0.06592, 0.79213, 0.08619) -- cycle {};
\fill[facet] (0.79467, -0.18679, -0.25936) -- (0.79467, -0.18679, 0.26521) -- (0.79467, -0.08963, 0.31295) -- (0.79467, 0.33304, 0.08619) -- (0.79467, 0.33304, -0.08429) -- (0.79467, -0.10859, -0.30131) -- cycle {};
\draw[edge] (0.79467, 0.33304, 0.08619) -- (0.79467, 0.33304, -0.08429);
\draw[edge] (0.79467, 0.33304, 0.08619) -- (0.79467, -0.08963, 0.31295);
\draw[edge] (0.79467, 0.33304, 0.08619) -- (0.06592, 0.79213, 0.08619);
\draw[edge] (0.79467, 0.33304, -0.08429) -- (0.79467, -0.10859, -0.30131);
\draw[edge] (0.79467, 0.33304, -0.08429) -- (0.06592, 0.79213, -0.08429);
\draw[edge] (0.79467, -0.08963, 0.31295) -- (0.79467, -0.18679, 0.26521);
\draw[edge] (0.79467, -0.08963, 0.31295) -- (-0.04744, -0.27956, 0.69946);
\draw[edge] (0.79467, -0.18679, 0.26521) -- (0.79467, -0.18679, -0.25936);
\draw[edge] (0.79467, -0.18679, 0.26521) -- (-0.04744, -0.37672, 0.65172);
\draw[edge] (0.79467, -0.10859, -0.30131) -- (0.79467, -0.18679, -0.25936);
\draw[edge] (0.06592, 0.79213, 0.08619) -- (0.06592, 0.79213, -0.08429);
\draw[edge] (0.06592, 0.79213, 0.08619) -- (-0.20777, 0.73041, 0.21181);
\draw[edge] (0.06592, 0.79213, -0.08429) -- (-0.22005, 0.72764, -0.20451);
\draw[edge] (-0.20777, 0.73041, 0.21181) -- (-0.27068, 0.71622, 0.18536);
\draw[edge] (-0.20777, 0.73041, 0.21181) -- (-0.20777, -0.17856, 0.69946);
\draw[edge] (-0.27068, 0.71622, 0.18536) -- (-0.27068, 0.71622, -0.18128);
\draw[edge] (-0.27068, 0.71622, 0.18536) -- (-0.27068, -0.19275, 0.67301);
\draw[edge] (-0.22005, 0.72764, -0.20451) -- (-0.27068, 0.71622, -0.18128);
\draw[edge] (-0.04744, -0.27956, 0.69946) -- (-0.04744, -0.37672, 0.65172);
\draw[edge] (-0.04744, -0.27956, 0.69946) -- (-0.20777, -0.17856, 0.69946);
\draw[edge] (-0.04744, -0.37672, 0.65172) -- (-0.24102, -0.42038, 0.57034);
\draw[edge] (-0.20777, -0.17856, 0.69946) -- (-0.27068, -0.19275, 0.67301);
\draw[edge] (-0.27068, -0.19275, 0.67301) -- (-0.27068, -0.40170, 0.57034);
\draw[edge] (-0.24102, -0.42038, 0.57034) -- (-0.27068, -0.40170, 0.57034);
\node[vertex] at (0.79467, 0.33304, 0.08619)     {};
\node[vertex] at (0.79467, 0.33304, -0.08429)     {};
\node[vertex] at (0.79467, -0.08963, 0.31295)     {};
\node[vertex] at (0.79467, -0.18679, 0.26521)     {};
\node[vertex] at (0.79467, -0.10859, -0.30131)     {};
\node[vertex] at (0.79467, -0.18679, -0.25936)     {};
\node[vertex] at (0.06592, 0.79213, 0.08619)     {};
\node[vertex] at (0.06592, 0.79213, -0.08429)     {};
\node[vertex] at (-0.20777, 0.73041, 0.21181)     {};
\node[vertex] at (-0.27068, 0.71622, 0.18536)     {};
\node[vertex] at (-0.22005, 0.72764, -0.20451)     {};
\node[vertex] at (-0.27068, 0.71622, -0.18128)     {};
\node[vertex] at (-0.04744, -0.27956, 0.69946)     {};
\node[vertex] at (-0.04744, -0.37672, 0.65172)     {};
\node[vertex] at (-0.20777, -0.17856, 0.69946)     {};
\node[vertex] at (-0.27068, -0.19275, 0.67301)     {};
\node[vertex] at (-0.24102, -0.42038, 0.57034)     {};
\node[vertex] at (-0.27068, -0.40170, 0.57034)     {};
\end{tikzpicture}
        \caption{$\ndoc=4$. The DBN-expohedron in $\R^2$ for a relevance vector $\vrho=(0.9, 0.7, 0.5, 0.1)$}
    \end{subfigure}
    \caption{Examples of DBN-expohedra for $\gamma=0.5,\kappa=0.7$.}
    \label{fig:dbn_expohedron}
\end{teaserfigure}
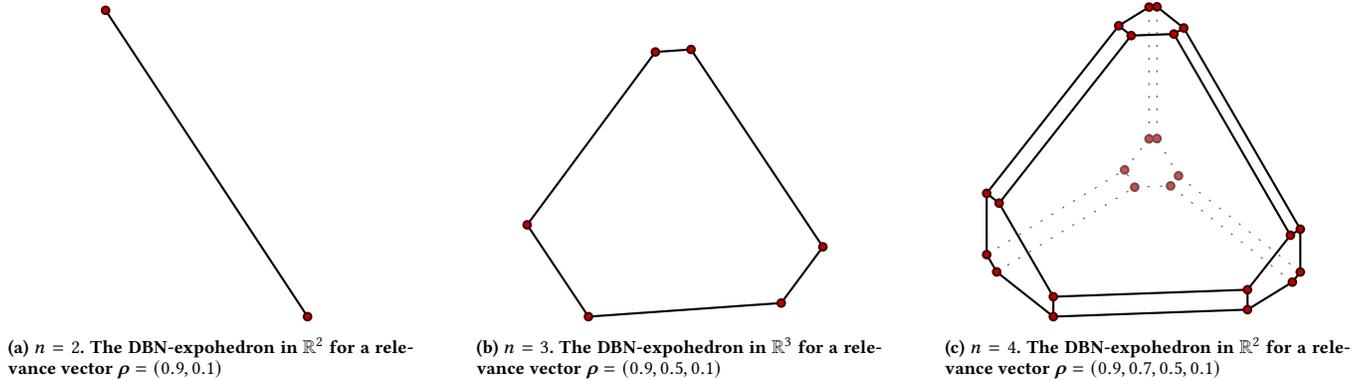

\begin{abstract}
    
In recent years, it has become clear that rankings delivered in many areas need not only be useful to the users but also respect fairness of exposure for the item producers. We consider the problem of finding ranking policies that achieve a Pareto-optimal tradeoff between these two aspects. Several methods were proposed to solve it; for instance a popular one is to use linear programming with a Birkhoff-von Neumann decomposition. These methods, however, are based on a classical Position Based exposure Model (PBM), which assumes independence between the items (hence the exposure only depends on the rank). In many applications, this assumption is unrealistic and the community increasingly moves towards considering other models that include dependences, such as the Dynamic Bayesian Network (DBN) exposure model. For such models, computing (exact) optimal fair ranking policies remains an open question.

In this paper, we answer this question by leveraging a new geometrical method based on the so-called expohedron proposed recently for the PBM (Kletti et al., WSDM'22). We lay out the structure of a new geometrical object (the DBN-expohedron), and propose for it a Carathéodory decomposition algorithm of complexity $O(\ndoc^3)$, where $\ndoc$ is the number of documents to rank. Such an algorithm enables expressing any feasible expected exposure vector as a distribution over at most $\ndoc$ rankings; furthermore we show that we can compute the whole set of Pareto-optimal expected exposure vectors with the same complexity $O(\ndoc^3)$. Our work constitutes the first exact algorithm able to efficiently find a Pareto-optimal distribution of rankings. It is applicable to a broad range of fairness notions, including classical notions of meritocratic and demographic fairness. We empirically evaluate our method on the TREC2020 and MSLR datasets and compare it to several baselines in terms of Pareto-optimality and speed.

\end{abstract}


\maketitle

\section{Introduction}

Automatic ranking systems take an increasingly large role in our everyday lives, be it as the result of a simple web search, online job markets, or receiving recommendations of new songs to listen to. 
Such systems connect item producers with item consumers and determine which items (webpages, songs, etc.) the consumers are exposed to, and which items remain hidden from sight.
While it is generally desirable to give much visibility to items that are relevant to consumers, it is also important to ensure that no item producers are unfairly disadvantaged in terms of exposure.

Fairness has been intensely studied in many applications such as classification \cite{dwork_fairness_2012} or regression \cite{agarwal_fair_2019} tasks.
Fairness in ranking, however, presents some particular challenges that do not appear in classification or regression \cite{pitoura_fairness_2021,zehlike_fairness_2021}.
In particular since highly-ranked items are more visible, and consequently more often clicked than lowly-ranked ones, the commonly adopted policy of ordering the items by decreasing order of relevance, the \acrfull{prp}, is the optimal solution for consumer utility (under some assumptions) \cite{robertson_probability_1977}.
On the other hand, the \acrshort{prp} generally yields exposures (\acrshort{ie} visibilities) that are not adequately related to the items' \emph{merit}, which is a type of unfairness we wish to avoid.
Such an ``adequate'' relation between exposure and merit is called \emph{fairness of exposure} \cite{singh_fairness_2018}.
In this paper our ambition is not to decide which exposures are fair, but to develop methods making it possible to achieve any feasible fair exposure.
We will illustrate this on a commonly used notion of fairness, namely \emph{meritocratic} fairness, for which the merit of an item is proportionate to its relevance \cite{morik_controlling_2020}, but it also applies on other notions such as demographic fairness.

Methods able to mitigate unfairness in rankings include \emph{pre-processing}, the treatment of data before input into a ranking component; \emph{in-processing}, the training of fair ranking models; and \emph{post-processing}, a step that involves re-ranking a certain set of $\ndoc$ items, given a query.
Within the post-processing step it is possible to deliver different rankings for the same query, a process called \emph{amortization}, because the unfairness of the initial ranking can be compensated in subsequent ones.
This is usually done by finding an optimal distribution over rankings, from which rankings are sampled.

Most approaches aiming to find such an optimal distribution consider a \acrfull{pbm}, for which the exposure given to an item is assumed to only depend on its rank.
For instance it is possible to use bistochastic matrices to represent distributions over permutations and optimize utility with a fairness constraint as a linear program \cite{singh_fairness_2018}.
Conversely it is possible to optimize fairness under a utility constraint as a linear program \cite{biega_equity_2018}.
\citeauthor{morik_controlling_2020} proposed to use a controller to make sure that the exposure is proportional to a notion of merit \cite{morik_controlling_2020}.
More recently a novel geometric method using a polytope called \emph{expohedron} was proposed \cite{kletti_introducing_2022}.
All these methods rely on using a \acrshort{pbm} as exposure model.
In a \acrshort{pbm}, however, the items are assumed \emph{independent}, such that the exposure of an item depends only on its rank. This is not a realistic assumption because the users, while browsing the list of results, may be influenced by the documents they see to continue or abandon the browsing.

In our work we consider a different range of exposure models encompassed by the \acrfull{dbn} model \cite{chapelle_dynamic_2009}, an exposure model that takes into account the relevance values from previously ranked items.
As such the \acrshort{dbn}, and its variants the \acrfull{cm}, \acrfull{dcm}, \acrfull{sdbn} and \acrfull{ccm} \cite{chuklin_click_2015}, are much more realistic browsing models.
Yet they have received much less attention from the fair ranking community, in part because these models are more complex than simple \acrshort{pbm}s.
Indeed, using a \acrshort{dbn} removes linearity from the problem, thereby making it impossible to solve with linear programming approaches such as \cite{singh_fairness_2018}.
The \acrshort{dbn} has been the subject of recent interest, in particular in the TREC2020 Fair Ranking track \cite{trec-fair-ranking-2019} and in a paper proposing a controller \cite{thonet_multi-grouping_2020} for it; but to date there is no method that provably computes an optimal distribution over rankings.

In this paper we draw upon the recent approach using the so-called \emph{expohedron}, proposed by  \cite{kletti_introducing_2022} for a \acrshort{pbm} exposure model, to work out a method able to find distributions over rankings with optimal expected exposures for a \acrshort{dbn} exposure model.
As \cite{kletti_introducing_2022} strongly relies on the particular (simple) structure of the \acrshort{pbm}, this entails a number of challenges, in particular to characterize theoretically the geometrical structure of the new \acrshort{dbn}-\emph{expohedron} that we obtain, illustrated in \cref{fig:dbn_expohedron}.
%
Then, however, it makes it possible, for any feasible target exposure, to find a distribution over rankings whose expected exposure is exactly the target exposure. Furthermore, it also enables to adapt the procedure in \cite{kletti_introducing_2022} to find the Pareto-front of a multi-objective fairness-utility optimization problem.

Our contributions can be summarized as follows:
\begin{trivlist}
    \item[(1)]\label{enum:geometry}
    Drawing upon the ideas of \cite{kletti_introducing_2022}, we define a new expohedron called \emph{\acrshort{dbn}-expohedron} based on the \acrshort{dbn} model and we provide theoretical results concerning the geometry of this new more complicated structure.
    \item[(2)] Using the above results, we devise an algorithm capable of finding the Carathéodory decomposition of any point inside a \acrshort{dbn}-expohedron.
    This makes it possible to express any feasible target exposure as the expected exposure of a distribution over rankings.
    This algorithm has complexity $O(\ndoc^3)$, where $\ndoc$ is the number of items to rank.
    \item[(3)] Using the above results we give an algorithm capable of finding the Pareto-front of a \acrfull{moo} problem invovling a fairness objective and a utility objective.
    This algorithm has complexity $O(\ndoc^3)$, where $\ndoc$ is the number of items to rank.
    \item[(4)] We perform experiments on the TREC2020 Fairness Track dataset \cite{trec-fair-ranking-2019} and on the MSLR dataset \cite{DBLP:journals/corr/QinL13}, and compare our approach in terms of Pareto-optimality and speed with respect to 3 baselines: A \acrfull{pl} policy \cite{diaz_evaluating_2020}, a heuristic controller \cite{thonet_multi-grouping_2020} and a \acrfull{ltr} method \cite{oosterhuis_computationally_2021}.
\end{trivlist}



\section{Related Work}\label{sec:rw}
    In this section we review closely related work and provide some detailed background that will be used throughout the paper.

\paragraph{Learning Plackett-Luce policies}

In a recent paper \cite{oosterhuis_computationally_2021}, \citeauthor{oosterhuis_computationally_2021} proposed a \acrfull{ltr} approach using \acrfull{pl} \cite{plackett_analysis_1975} distributions as ranking policy. More specifically, the rankings are generated by a \acrfull{pl} distribution whose log-scores are given by a trained model using features associated to the (query,document) pairs as inputs.

The loss function used to train this model by policy gradient can be any convex combination of utility and fairness, which makes it possible to choose a loss function that is a scalarization of a multi-objective optimization problem.
In their paper, \cite{oosterhuis_computationally_2021} consider a \acrshort{pbm} for the exposure, but the approach can be adapted to a \acrshort{dbn}, without great difficulty.

\paragraph{Optimization over bistochastic matrices}
A popular method to find optimal distributions over rankings is to represent them as bistochastic matrices \cite{singh_fairness_2018,singh_fairness_2021,wang_user_2021}.
The so-called \acrfull{bvn} algorithm can express any bistochastic matrix as the expected value of a distribution over permutations matrices \cite{dufosse_notes_2016}.
When a \acrshort{pbm} is used, this allows to run a linear program to find an optimal bistochastic matrix, because with a \acrshort{pbm}, the exposure of a ranking is a linear function of the bistochastic matrix.
When a \acrshort{dbn} exposure model is used, such a linear relationship does not exists and using bistochastic matrices no longer makes sense.
Indeed the elements of a bistochastic matrix represent the marginal probabilities for the items of ending up at a certain rank; but as we will show later (see \eqref{eq:dbn}), this marginal information does not suffice to deduce the expected exposure of an item.
Therefore another method is needed.

\paragraph{Expohedron}
Recently, \citeauthor{kletti_introducing_2022} \cite{kletti_introducing_2022} introduced a polytope called \emph{expohedron}, whose vertices represent all exposure vectors attainable with one ranking with a \acrshort{pbm} and whose interior is the convex hull of the vertices.
The points of this polytope represent all expected exposure vectors of distributions over rankings.
The authors propose an algorithm capable of expressing any point inside the expohedron as the expected exposure of a distribution over permutations; this is called a \emph{Carathéodory decomposition}.
Furthermore they propose an algorithm able to find the set of all Pareto-optimal exposure vectors in the expohedron.
Both these algorithms have complexity $O(\ndoc^2\log(\ndoc))$, where $\ndoc$ is the number of items to rank.
However their approach strongly relies on the particular structure of \acrshort{pbm}, which does not naturally extend to the \acrshort{dbn} exposure model.

\paragraph{Controller}
In \citeyear{thonet_multi-grouping_2020}, \citeauthor{thonet_multi-grouping_2020} \cite{thonet_multi-grouping_2020} proposed a heuristic controller that starts by delivering a \acrshort{prp} ranking, looks at how much each item is advantaged or disadvantaged \acrshort{wrt} its merit and corrects this in the subsequent ranking by artificially increasing or decreasing the relevance values of the items and delivering the next ranking by ordering the items according to the corrected relevance values.
This method is compatible with a \acrshort{dbn} exposure method and will be used as a baseline.


\section{Model}\label{sec:model}
    \paragraph{Setting}
We assume a single or anonymous user that issues a query $q$ many times and we are not preoccupied about being fair towards users, but only to the provider side.
We suppose that for such a query $q$ we are provided $\ndoc$ items and estimations of their relevance value \acrshort{wrt} the query, between $0$ (irrelevant) and $1$ (relevant).
We suppose that those relevance estimations are unbiased and represent the true relevances.
This is obviously an unrealistic assumption, but it has the merit of enabling us to study the amortization problem independently of the problem of fairly estimating relevance values.
Given a query $q$, our task is to deliver a sequence of rankings that is optimal in terms of utility provided to the user and in terms of fairness of item exposure.
The remainder of this section is dedicated to the formal definitions of ranking, exposure, user-side utility, item-side fairness and to the formal introduction of our optimization problem.

\paragraph{Ranking}
We formalize a ranking as a permutation $\pi\in\cS_\ndoc$ of $\ndoc$ items such that if item $\docix$ is at rank $\rankix$ we have $\docix=\pi(\rankix)$.
Here $\cS_\ndoc$ represents the set of permutations of size $\ndoc$.
Given a relevance vector $\vrho\in\R^\ndoc$, we denote $\piprp$ an arbitrary ranking such that
\begin{equation}
    \vrho_{\piprp(1)}\geq\hdots\geq\vrho_{\piprp(\ndoc)}.
\end{equation}
A ranking $\piprp$ orders the items by decreasing relevance values.

\paragraph{Exposure model}
We use the generic \acrfull{dbn} exposure model \cite{chapelle_dynamic_2009}.
Unlike with the simpler \acrfull{pbm} used in \cite{kletti_introducing_2022,oosterhuis_computationally_2021,singh_fairness_2018}, the exposures of a \acrshort{dbn} depend on the actual relevances of the items in the ranking, thereby making it a more realistic model.
The \acrshort{dbn} model has two parameters:  $\gamma\in[0,1]$  (the continuation probability) and $\kappa\in[0,1]$ (the satisfaction probability of a relevant item).
Furthermore it depends on the relevance vector $\vrho\in\R^\ndoc$.
Formally the expression of the exposure $\vcE_\docix$ of an item $\docix$ for a ranking $\pi\in\cS_\ndoc$ is
\begin{equation}\label{eq:dbn}
    \vcE_\docix(\pi,\gamma,\kappa,\vrho) \defeq \gamma^{\pi^{-1}(i)-1} \prod_{l=1}^{\pi^{-1}(i)-1} \left(1-\kappa \vrho_{\pi(l)}\right).
\end{equation}
Here $\pi^{-1}(i)$ is the rank of item $\docix$.
Equation \eqref{eq:dbn} has an intuitive interpretation: the user runs through the list and at each rank $l$, the user has a probability $1-\gamma$ of simply abandoning his search and a probability $\kappa\vrho_{\pi(l)}$ of being satisfied with the item $\pi(l)$ and as a result stopping their search.
In this paper we often omit the heavy notation involving $\gamma,\kappa,\vrho$ and write the exposure vector simply as $\vcE(\pi)=\left(\vcE_1(\pi),\hdots,\vcE_\ndoc(\pi)\right)^\top$.

The \acrshort{dbn} model contains as particular cases many other click models, such as the \acrlong{cm}, the \acrlong{sdbn} model, the \acrlong{dcm} and the \acrlong{ccm} \cite{chuklin_click_2015}, see Appendix \ref{app:click}.

Given a distribution over rankings $\cD$, the exposure vector of $\cD$ is the expectation
\begin{equation}
    \vcE(\cD) \defeq \E_{\pi\sim\cD}[\vcE(\pi)].
\end{equation}

\paragraph{User side utility}

We define the utility of a ranking as the scalar product of exposure with relevance, as this corresponds to the \acrfull{err} metric \cite{chapelle_expected_2009} that was developed jointly with the \acrshort{dbn} model \cite{chapelle_dynamic_2009}.
Formally we denote

\begin{equation}
    \utility(\pi) \defeq \vrho^\top\vcE(\pi),
\end{equation}
where $\cdot^\top$ is the transpose operator.
Given a distribution over rankings $\cD$, the utility of $\cD$ is
\begin{equation}\label{eq:utility}
    \utility(\cD) = \E_{\pi\sim\cD}[\utility(\pi)].
\end{equation}

In particular the utility is a linear function of the exposure vector.
In fact this definition is analogous to the popular \acrshort{dcg} metric \cite{jarvelin_cumulated_2002} with the difference that a \acrshort{pbm} exposure is used for the \acrshort{dcg}.
Similarly to the relationship that \acrshort{dcg} bears with \acrshort{ndcg}, we define the \defkey{normalized utility} and denote $\nutility$ the utility divided by the utility obtained with a \acrshort{prp} ranking:
\begin{equation}
    \nutility(\pi) \defeq \frac{\utility(\pi)}{\utility(\piprp)}.
\end{equation}
This gives an adimensional metric that is useful to be able to aggregate utilities across different queries in a meaningful way.

\paragraph{(Un)fairness}
We assume that a decision-maker has decided upon a vector of merits $\vmu\in\R^\ndoc$ with non-negative components.
The fact that the vector of merits is a free parameter makes our setting very flexible in terms of covered fairness notions.
This vector of merits can be set to be equal to the relevance vector so as to get a notion of \emph{meritocratic} fairness, or it can be set to $\vone$, so as to get a notion of \emph{demographic} fairness.
An exposure vector $\vcE$ is said to be fair if it is proportional to the vector of merits $\vmu$, \acrshort{ie} if there exists a positive real number $k\in\R_+$ such that $\vcE=k\vmu$.
As we will see in \cref{sec:expohedron}, there exists at most one exposure vector proportional to $\vmu$ such that it is achievable by a distribution over rankings, \acrshort{ie} such that there exists a distribution over rankings $\cD$ with $\vcE(\cD)\propto\vmu$.
When such a feasible vector exists, we denote it by $\vcEt$ and call it the \defkey{target exposure}.
When such a vector does not exist, we relax the proportionality relationship into an affine one, by adding a constant value to every element of $\vmu$ until the corresponding vector is feasible, \acrshort{ie} we use the merit vector $\vmu'=\vmu+K\vone$ for minimal $K$, and we define the target exposure as $\vcEt\propto\vmu'$.
This definition of target exposure is to a certain degree arbitrary and we acknowledge that other definitions are possible such as projecting $\vmu$ on the expohedron with a euclidean distance.

We measure the unfairness of a distribution over rankings as the euclidean distance of its expected exposure to the target exposure:

\begin{equation}\label{eq:unfairness}
    \unfairness(\cD) = \norm{\E_{\pi\sim\cD}[\vcE(\pi)] - \vcEt}_2.
\end{equation}
As for the normalized utility, we define the \defkey{normalized unfairness} as the unfairness divided by the unfairness obtained with a \acrshort{prp} ranking:
\begin{equation}\label{eq:nunfairness}
    \nunfairness(\cD) \defeq \frac{\norm{\E_{\pi\sim\cD}[\vcE(\pi)] - \vcEt}_2}{\norm{\vcE_{\PRP} - \vcEt}_2}.
\end{equation}
This normalization has the advantage of giving a metric with values between $0$ and $1$ (for Pareto-optimal distributions) independently of the number of documents $\ndoc$
Equation \eqref{eq:nunfairness} assumes that the target exposure is always different from the exposure obtained with a \acrshort{prp} ranking.
This is not a very restrictive assumption, because when the two exposure vectors coincide, the problem becomes trivial, as it suffices to deliver a \acrshort{prp} ranking to get both maximal utility and minimal unfairness.

\paragraph{Optimization problem}

The normalized unfairness and the normalized utility are the two objectives of the \acrfull{moo} problem:
\begin{equation}\label{eq:moo}
    \max_{\cD} \nutility(\cD),\quad\min_{\cD} \nunfairness(\cD).
\end{equation}

Our goal is to find a set of distributions $\cD$ of rankings that is Pareto-optimal for the two objectives of \eqref{eq:moo}.
Note that in \eqref{eq:utility} and \eqref{eq:nunfairness}, $\nunfairness$ and $\nutility$ depend on $\cD$ only through the expected value $\vcE(\cD)$.
Therefore it is possible to decompose the \acrshort{moo} \eqref{eq:moo} into two sub-problems:
\begin{enumerate}
    \item Finding all Pareto optimal vectors $\vcE\in\R^\ndoc$ that are the expectation of some distribution $\cD$,
    \item Given a Pareto-optimal exposure vector $\vcE\in\R^\ndoc$, finding a distribution $\cD$ such that $\vcE=\E_{\pi\sim\cD}[\vcE(\pi)]$.
\end{enumerate}

In the remainder of the paper we draw on the ideas developed by \citeauthor{kletti_introducing_2022} \cite{kletti_introducing_2022}, see \cref{sec:rw}.
We first derive several properties of the so-called expohedron resulting from a \acrshort{dbn} exposure model.
We show that these properties make it possible to design a Carathéodory decomposition algorithm in the \acrshort{dbn}-expohedron, thereby making it possible to find a distribution $\cD$ such that $\vcE(\cD)=\vcE$ for any $\vcE$ for which such a distribution exists.
Then we show that using these properties we can recover the whole set of Pareto-optimal exposure vectors for the \acrshort{moo} \eqref{eq:moo}.
Both these algorithms have complexity $O(\ndoc^3)$, where $\ndoc$ is the number of documents to be ranked.

\section{The DBN-expohedron}\label{sec:expohedron}
    In this section we formally define the \acrshort{dbn}-expohedron and provide several properties about its geometry.
In previous work \cite{kletti_introducing_2022}, the expohedron has been defined as the convex hull of the exposure vectors achieved with a ranking, using a \acrshort{pbm}.
In this paper we use the same definition, but using a \acrshort{dbn} exposure model instead of a \acrshort{pbm}.

\paragraph{The \acrshort{dbn}-expohedron}
Given a \acrshort{dbn} exposure model parametrized by $\gamma\in[0,1),~\kappa\in[0,1]$ and given an exposure vector $\vrho\in[0,1]^\ndoc$, we define the \acrshort{dbn}-expohedron as
\begin{equation}
    \Pi(\gamma,\kappa,\vrho) \defeq \Conv\left(\left\{\vcE(\pi,\gamma,\kappa,\vrho)~|~\pi\in\cS_\ndoc\right\}\right).
\end{equation}
The \acrshort{dbn}-expohedron is the convex hull of all exposure vectors $\vcE(\pi)$ with $\pi\in\cS_\ndoc$ and $\cS_\ndoc$ is the set of permutations of size $\ndoc$.
As such it contains exactly those exposure vectors that are expected values of distributions over rankings.
This is why we say that vectors inside the expohedron are \emph{feasible}.

\paragraph{Properties}

A first important property is that the \acrshort{dbn}-expohedron is actually contained in a hyperplane of $\R^\ndoc$, \acrshort{ie} it is an object of dimension $\ndoc-1$.

\begin{theoremEnd}{prop}\label{prop:normal}
            The polytope $\Pi(\gamma,\kappa,\vrho)$ is contained in a hyperplane with normal vector $\vnu = \vone+\frac{\gamma\kappa}{1-\gamma}\vrho$.
\end{theoremEnd}
\begin{proofEnd}
    The idea is to show that the scalar product
    \begin{equation}\label{eq:nuE}
        \vnu^\top\vcE(\pi) 
    \end{equation}
    does note change whenever the ranks of two items of adjacent ranks are exchanged.
    All rankings can be obtained by successively exchanging the ranks of adjacent items, starting from any initial ranking.
    Therefore it suffices to show that \eqref{eq:nuE} remains unchanged by the exchange of the ranks of two adjacent elements.
    Given an arbitrary ranking $\pi_0$ the exposures of ranks $\rankix$ and $\rankix+1$ are
    \begin{equation}
        \hdots,\underbrace{K}_{\text{rank } \rankix,\text{ item }\pi_0(\rankix)}, \underbrace{K\gamma(1-\kappa\vrho_{\pi_0(\rankix)})}_{\text{rank } \rankix+1,\text{ item }\pi_0(\rankix+1)},\hdots
    \end{equation}
    for a certain constant $K$.
    When the items at ranks $\rankix$ and $\rankix+1$ are exchanged, to give a new ranking $\pi_1$, these exposures become
    \begin{equation}
        \hdots,\underbrace{K}_{\text{rank }\rankix,\text{ item }\pi_0(\rankix+1)}, \underbrace{K\gamma(1-\kappa\vrho_{\pi_0(\rankix+1)})}_{\text{rank }\rankix+1,\text{ item }\pi_0(\rankix)},\hdots
    \end{equation}
    while the exposures at all other ranks remain unchanged.
    Let us denote $i=\pi_0(\rankix)$ and $j=\pi_0(\rankix+1)$.
    With some algebraic manipulations, it is then easy to show that the scalar products \eqref{eq:nuE} for rankings $\pi_0$ and $\pi_1$ satisfy
        \begin{align}
            \begin{split}
                \vnu^\top\vcE(\pi_1) - \vnu^\top\vcE(\pi_0) &= \left(1+\frac{\gamma\kappa}{1-\gamma}\vrho_i\right)\left(K\gamma(1-\kappa\vrho_j) - K\right)\\
                &\phantom{=}+ \left(1+\frac{\gamma\kappa}{1-\gamma}\vrho_j\right)\left(K-K\gamma(1-\kappa\vrho_i)\right)\\
            &=0.
            \end{split}
        \end{align}
    This means that \eqref{eq:nuE} remains unchanged by exchanging two items of adjacent ranks, and thus that \eqref{eq:nuE} is constant in the whole expohedron.
\end{proofEnd}

From \cref{prop:normal}, it follows that for every vector of merits $\vmu$, there exists at most one feasible expected exposure $\vcEt$ vector such that $\vcEt\propto\vmu$, because the arrow in the direction $\vmu$ intersects the hyperplane at most once.

\begin{rema}
    The fact that the normal vector is not $\vone$, as would be the case for a \acrshort{pbm} expohedron \cite{kletti_introducing_2022}, implies that minimizing $\norm{\vcE}$, as in \cite{diaz_evaluating_2020}, does not lead to equal exposures, as is illustrated in \Cref{fig:expohedron}.
\end{rema}

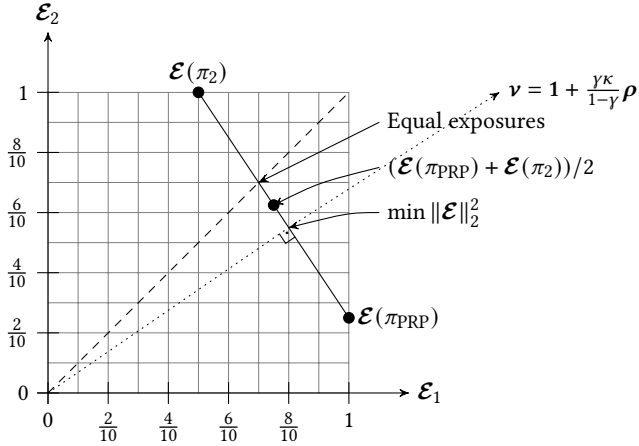
\begin{figure}
    \centering
    \begin{tikzpicture}[scale=4,cap=round]
  \def\costhirty{0.8660256}

  \colorlet{anglecolor}{green!50!black}
  \colorlet{sincolor}{red}
  \colorlet{tancolor}{orange!80!black}
  \colorlet{coscolor}{blue}

  \tikzstyle{axes}=[]
  \tikzstyle{important line}=[very thick]
  \tikzstyle{information text}=[rounded corners,fill=red!10,inner sep=1ex]

  \draw[very thin, step=0.1cm, gray] (0,0) grid (1,1);

  \begin{scope}[style=axes]
    \draw[->] (-0.1,0) -- (1.2,0) node[right] {$\vcE_1$};
    \draw[->] (0,-0.1) -- (0,1.2) node[above] {$\vcE_2$};

    \foreach \x/\xtext in {0, 0.2/\frac{2}{10}, 0.4/\frac{4}{10}, 0.6/\frac{6}{10}, 0.8/\frac{8}{10}, 1}
      \draw[xshift=\x cm] (0pt,1pt) -- (0pt,-1pt) node[below,fill=white]
            {$\xtext$};

    \foreach \y/\ytext in {0, 0.2/\frac{2}{10}, 0.4/\frac{4}{10}, 0.6/\frac{6}{10}, 0.8/\frac{8}{10}, 1}
      \draw[yshift=\y cm] (1pt,0pt) -- (-1pt,0pt) node[left,fill=white]
            {$\ytext$};
  \end{scope}

	\draw[dashed] (0, 0) -- (1, 1);
	\draw[dotted,->] (0, 0) coordinate (O) -- (8/10, 11/20) coordinate (X) -- (1.5,1) node[align=left, right] {$\vnu=\vone+\frac{\gamma\kappa}{1-\gamma}\vrho$};
	\draw[<-] (8/10, 11/20) .. controls (1,0.6) .. (1.1,0.6) node[align=left, right] {$\min\norm{\vcE}_2^2$};
	\draw[<-] (7/10, 7/10) .. controls (0.9,0.8) .. (1.1,0.9) node[align=left, right] {Equal exposures};

	\filldraw (1,0.25) coordinate (A) circle (0.5pt) node[align=left, right] {$\vcE(\pi_{\PRP})$} --
	(0.5,1) circle (0.5pt) node[align=left, above] {$\vcE(\pi_2)$};
	\filldraw (0.75, 1.25/2) coordinate (M) circle (0.5pt);
	\draw[<-] (M) .. controls (0.9,0.7) .. (1.1,3/4) node[align=left, right] {$(\vcE(\pi_{\PRP}) + \vcE(\pi_2)) / 2$};

	\draw pic [draw=black,angle eccentricity=1,pic text=$\cdot$,scale=0.3] {right angle = O--X--A};

\end{tikzpicture}
    \caption{A \acrshort{dbn}-expohedron for $\ndoc=2$ documents of relevances $\vrho=(1, 0)$ with parameters $\gamma=\kappa=0.5$.
    It is apparent that minimizing $\norm{\vcE}^2_2$ does not lead to any interpretable notion of equity.}
    \label{fig:expohedron}
\end{figure}

Before further characterizing the faces of our expohedron, let us recall the definition of \emph{zone} $\zone(\pi)$ from \cite{kletti_introducing_2022} as the subset of $\R^\ndoc$ of vectors $\vx$ such that
\begin{equation}
    \vx_{\pi(1)} \geq \hdots \geq \vx_{\pi(\ndoc)}.
\end{equation}

The following proposition gives a characterization of the faces of the \acrshort{dbn}-expohedron that is similar to the characterization in \cite{kletti_introducing_2022}.
However the normal vectors to the faces take a very different expression with a \acrshort{dbn} model.
\begin{theoremEnd}{prop}\label{prop:faces}
    Every face $F$ of $\Pi(\gamma,\kappa,\vrho)$ is characterized by
    \begin{enumerate}
        \item A zone $\zone(\pi)$ with $\pi\in\cS_\ndoc$,
        \item A subset $\splits$ of $\seq{1}{\ndoc}$ called set of \defkey{splits}.
    \end{enumerate}
    We denote a face by $F=(\pi,\splits)$.
    In the basis in which $\pi$ is the identity, the $|\splits|$ normal vectors of a face are given by
    \begin{equation}
        \vnu_s=\vone(s)+\frac{\gamma\kappa}{1-\gamma}(\vrho\odot\vone(s)), \quad\forall s\in\splits,
    \end{equation}
    where
    \begin{equation}
        \vone(s)_{\docix} =
        \begin{cases}
        1\text{, if }\docix \leq s,\\
        0\text{, else },
        \end{cases}
    \end{equation}
    and where $\odot$ is the element-wise product.
    Furthermore the dimension of a face $F=(\pi,\splits)$ is $\dim(F) = \ndoc - |\splits|$.
\end{theoremEnd}
\begin{proofEnd}
    The proof to this theorem is quite complex, we first introduce two lemmas to build an equivalence with a well-known polytope called \emph{permutohedron}, then we show how the result can be obtained as a corollary of \cref{prop:normal}.
    \begin{lemm}
        For every $\pi\in\cS_\ndoc$, the point $\vcE(\pi,\gamma,\kappa,\vrho)$ is in the zone $\zone(\pi)$.
    \end{lemm}
    \begin{proof}
        We need to show that
        \begin{equation}
            \vcE_{\pi(1)} \geq\hdots\geq \vcE_{\pi(\ndoc)}.
        \end{equation}
        This is clearly the case, because the exposure decreases with increasing rank.
    \end{proof}
    \begin{lemm}
        The vertices of the the \acrshort{dbn}-expohedron $\Pi(\gamma,\kappa,\vrho)$ are exactly the $\ndoc!$ points $\vcE(\pi,\gamma,\kappa,\vrho)$ for $\pi\in\cS_\ndoc$.
        Formally
        \begin{equation}
            \vx\in\Pi(\gamma,\kappa,\vrho)\text{ is a vertex } \iff \exists\pi\in\cS_\ndoc:~ \vx=\vcE(\pi,\gamma,\kappa,\vrho).
        \end{equation}
    \end{lemm}
    \begin{proof}~\\
        \begin{itemize}
            \item "$\implies$" Every vertex of $\Pi(\gamma,\kappa,\vrho)$ is a point of type $\vcE(\pi,\gamma,\kappa,\vrho)$, because $\Pi(\gamma,\kappa,\vrho)$ is the convex hull of these points.
            \item "$\impliedby$" To prove that a point in a polytope is a vertex, it suffices to show that it cannot be expressed as a convex combination of other vertices.
            Let $\pi\in\cS_\ndoc$ such that
            \begin{equation}
                \vcE(\pi,\gamma,\kappa,\vrho) = \sum_{i=1}^N \alpha_i\vcE(\pi_i,\gamma,\kappa,\vrho),
            \end{equation}
            with $\sum_{i=1}^N\alpha_i=1$ and $\alpha_i>0$ for all $i\in\seq{1}{N}$.
            One can show by recursion over $k\in\seq{1}{\ndoc}$, that
            \begin{equation}
                \vcE_{\pi(k)}(\pi_i,\gamma,\kappa,\vrho) = \vcE_{\pi(k)}(\pi,\gamma,\kappa,\vrho),~\forall i\in\seq{1}{N}.
            \end{equation}
            The highest component of $\vcE(\pi,\gamma,\kappa,\vrho)$ is the number $1$ at index $\pi(1)$, so all other $\vcE(\pi_i,\gamma,\kappa,\vrho)$ must also have a $1$ at index $\pi(1)$.
            If all $\vcE(\pi_i,\gamma,\kappa,\vrho)$ have equal components at indices $\pi(1),\hdots,\pi(k)$, then by the same reasoning, they must have equal components at indices $\pi(k+1)$.
            By recursion, they have equal indices at all components, which concludes the proof.
        \end{itemize}
    \end{proof}
    From previous lemmas we can conclude that $\Pi(\gamma,\kappa,\vrho)$ has the same combinatorial structure as the \acrshort{pbm}-expohedron \cite{kletti_introducing_2022} and as the permutohedron.
    Therefore a face of $\Pi(\gamma,\kappa,\vrho)$ is composed of vertices defined by a partial ordering of the items.
    Such a partial ordering can be defined by a permutation $\pi\in\cS_\ndoc$ and a set of splits $\splits\subseteq\seq{1}{\ndoc}$.
    Given a permutation $\pi\in\cS_\ndoc$ and a set of splits $\{s_1,\hdots,s_N\}\subseteq\seq{1}{\ndoc}$ with $s_N=\ndoc$, the partial ordering is such that, when working in the basis of $\R^\ndoc$ in which $\pi$ is the identity, we have
    \begin{equation}
        \vcE(\pi')_1,\hdots,\vcE(\pi')_{s_1} \geq \hdots \geq \vcE(\pi')_{s_{N-1}+1},\hdots,\vcE(\pi')_{s_N},
    \end{equation}
    for every vertex $\vcE(\pi')$ of the face.
    In particular, if we ignore the ranks $s_{k}+1,\hdots,\ndoc$ for $k\in\seq{1}{N}$, we can apply \cref{prop:normal} to the ranks $s_1,\hdots,s_{k}$ and get
    \begin{equation}
        \vnu_{s_k}=\vone(s_k)+\frac{\gamma\kappa}{1-\gamma}(\vrho\odot\vone(s_k))
    \end{equation}
    as a normal vector to the face, which is what there was to prove.
\end{proofEnd}

With \cref{prop:faces} it becomes possible to check whether an arbitrary point $\vx\in\R^\ndoc$ is inside $\Pi(\gamma,\kappa,\vrho)$ using \cref{alg:is_inside}.
This non-trivial result is necessary in \cref{sec:methods} to design our Carathéodory decomposition and to find the Pareto-front of our \acrshort{moo} problem.

\begin{algorithm}
	\caption{Check whether a point $\vx\in\R^\ndoc$ is inside $\Pi(\gamma,\kappa,\vrho)$}\label{alg:is_inside}
	\begin{algorithmic}[1]
		\Procedure{\texttt{is\_inside}}{\textsc{Input}: $\vx\in\R^\ndoc,\Pi(\gamma,\kappa,\vrho)$}
		\State $\pi \gets \argsort(\vx)$ \Comment{Identify the zone of $\vx$}
		\State $\vv \gets \vcE(\pi)$ \Comment{Create the zone's vertex}
		\State $\texttt{is\_inside} \gets \texttt{True}$
		\If{$\vnu^\top(\vx - \vv) \neq 0$}\label{l:check_expo}
		        \State $\texttt{is\_inside} \gets \texttt{False}$
		    \EndIf
		\For{$s\in\seq{1}{\ndoc-1}$}
		    \If{$\vnu_s^\top(\vx - \vv) > 0$}\label{l:check_facets}
		        \State $\texttt{is\_inside} \gets \texttt{False}$
		    \EndIf
		\EndFor
		\EndProcedure~\textsc{Output}: $\texttt{is\_inside}$
	\end{algorithmic}
\end{algorithm}
\begin{theoremEnd}{lemm}
    \Cref{alg:is_inside} returns \texttt{True} if and only if $\vx\in\Pi(\gamma,\kappa,\vrho)$.
    It operates with complexity $O(\ndoc^2)$.
\end{theoremEnd}
\begin{proofEnd}
    The expohedron can be partitioned into $\ndoc!$ pieces, each of which is contained by a zone.
    A point $\vx$ in a zone $\zone(\pi)$ is in the expohedron if and only if it is in the part of the expohedron contained by $\zone(\pi)$.
    Within a zone, the expohedron is delimited by $\ndoc$ hyperplanes: one hyperplane corresponds to the hyperplane of the whole expohedron (checked in line \ref{l:check_expo}) and $\ndoc-1$ other hyperplanes corespond to the $\ndoc-1$ facets adjacent to the only vertex of zone $\zone(\pi)$ (checked in line \ref{l:check_facets}).

    The determination of the zone of $\vx$ has complexity $O(\ndoc\log(\ndoc))$.
    In the loop over $s\in\seq{1}{\ndoc-1}$, each scalar product has complexity $O(\ndoc)$.
    Therefore the complexity of the whole algorithm is $O(\ndoc^2)$.
\end{proofEnd}

With \cref{prop:faces} it also becomes possible to identify the lowest-dimensional face a point of the expohedron is contained in, using \cref{alg:face_id}.

\begin{algorithm}
	\caption{Identify the smallest face in which a point is contained}\label{alg:face_id}
	\begin{algorithmic}[1]
		\Procedure{\texttt{face\_id}}{\textsc{Input}: $\vx\in\Pi(\gamma,\kappa,\vrho)$}
		\State $\pi \gets \argsort(\vx)$ \Comment{Identify the zone of $\vx$}
		\State $\vv \gets \vcE(\pi)$ \Comment{Create the zone's vertex}
		\State $\splits \gets \varnothing$
		\For{$s\in\seq{1}{\ndoc}$}
		    \If{$\vnu_s^\top(\vv-\vx) = 0$}
		        \State $\splits \gets \splits \cup \{s\}$
		    \EndIf
		\EndFor
		\EndProcedure~\textsc{Output}: The face $F = (\pi,\splits)$
	\end{algorithmic}
\end{algorithm}
\begin{theoremEnd}{lemm}
    \Cref{alg:face_id} finds the smallest face in which a point $\vx\in\Pi(\gamma,\kappa,\vrho)$ is contained.
    It operates with complexity $O(\ndoc^2)$.
\end{theoremEnd}
\begin{proofEnd}
    Once we know that a point $\vx$ is in the zone $\zone(\pi)$, it can be in at most $\ndoc-1$ facets of the expohedron at once, because there are $\ndoc-1$ facets adjacent to each unique vertex of a zone.
    Since every face of the expohedron is in the intersection of a certain number of facets belonging to the same zone, it suffices to check to which facets of zone $\zone(\pi)$, the point $\vx$ belongs in order to know to which faces he belongs.
    The lowest-dimensional such face is the one with the most normal vectors, which is what the algorithm finds.
    The complexity of the initial zone determination is $O(\ndoc\log(\ndoc)$, and the loop performs $\ndoc$ scalar products of vectors of size $\ndoc$, so the total complexity is $O(\ndoc^2)$.
\end{proofEnd}

Furthermore with \cref{prop:faces} it becomes possible to project vectors on the lowest-dimensional subspace containing a certain face $F=(\pi,\splits)$, since all normal vectors are known.
This can be done by building an orthonormal projection matrix from the set $\{\vnu_s~|~s\in\splits\}$ using for instance the Gram-Schmidt orthonormalization process.

Thanks to these results, we can design a Carathéodory decomposition algorithm in the expohedron and find the set of Pareto-optimal solutions to \eqref{eq:moo}, as we will see in \cref{sec:methods}.

\section{Our proposed algorithms}\label{sec:methods}
    \subsection{Carathéodory decomposition}
In order to be able to express any feasible target exposure as the expected exposure of a distribution over rankings, we need a method able to express any point inside the \acrshort{dbn}-expohedron as a convex combination of its vertices.
Such a combination that uses at most $\ndoc$ vertices is called a \emph{Carathéodory decomposition}, after a famous theorem of Carathéodory \cite{caratheodory_uber_1907,naszodi_perron_2019}.

In this section we expose such a method by showing how one can adapt the generic \acrshort{gls} method, named after Grötschel, Lovász and Schrijver \cite{grotschel_geometric_1993} to the \acrshort{dbn}-expohedron.
The \acrshort{gls} method is a generic Carathéodory decomposition algorithm that needs to be adapted to the particular polytope on which it is applied.
In recent work \cite{kletti_introducing_2022} the \acrshort{gls} method was adapted to the \acrshort{pbm}-expohedron, but does not naturally extend to the \acrshort{dbn}-expohedron, for which such a decomposition algorithm remains to be found.

The two steps in the \acrshort{gls} method that need adaptation are the following:
\begin{enumerate}
    \item\label{enum:intersection} We need to find the intersection of a half-line starting from inside the polytope, with the border of the polytope.
    \item\label{enum:vertex_determination} Given a point on a face of the polytope, we need to find a method able to find a vertex of the same face.
\end{enumerate}

\begin{figure}
    \centering

\begin{tikzpicture}%
	[x={(-0.697158cm, -0.389674cm)},
	y={(0.716917cm, -0.378888cm)},
	z={(-0.000040cm, 0.839404cm)},
	scale=3.000000,
	back/.style={loosely dotted, thin},
	edge/.style={color=black, thick},
	facet/.style={fill=white,fill opacity=0.300000},
	vertex/.style={inner sep=1pt,circle,draw=red!25!black,fill=red!75!black,thick}]
%
%

\coordinate (1.00000, 0.46500, 0.15113) at (1.00000, 0.46500, 0.15113);
\coordinate (1.00000, 0.08603, 0.46500) at (1.00000, 0.08603, 0.46500);
\coordinate (0.32500, 1.00000, 0.15113) at (0.32500, 1.00000, 0.15113);
\coordinate (0.06012, 1.00000, 0.32500) at (0.06012, 1.00000, 0.32500);
\coordinate (0.18500, 0.08603, 1.00000) at (0.18500, 0.08603, 1.00000);
\coordinate (0.06012, 0.18500, 1.00000) at (0.06012, 0.18500, 1.00000);

\coordinate (O) at (2, 2, 2);

\fill[facet] (0.06012, 0.18500, 1.00000) -- (0.06012, 1.00000, 0.32500) -- (0.32500, 1.00000, 0.15113) -- (1.00000, 0.46500, 0.15113) -- (1.00000, 0.08603, 0.46500) -- (0.18500, 0.08603, 1.00000) -- cycle {};
\draw[edge] (1.00000, 0.46500, 0.15113) -- (1.00000, 0.08603, 0.46500);
\draw[edge] (1.00000, 0.46500, 0.15113) -- (0.32500, 1.00000, 0.15113);
\draw[edge] (1.00000, 0.08603, 0.46500) -- (0.18500, 0.08603, 1.00000);
\draw[edge] (0.32500, 1.00000, 0.15113) -- (0.06012, 1.00000, 0.32500);
\draw[edge] (0.06012, 1.00000, 0.32500) -- (0.06012, 0.18500, 1.00000);
\draw[edge] (0.18500, 0.08603, 1.00000) -- (0.06012, 0.18500, 1.00000);
\node[vertex, label=left:$\vv_1$] (V1) at (1.00000, 0.46500, 0.15113)     {};
\node[vertex] at (1.00000, 0.08603, 0.46500)     {};
\node[vertex] at (0.32500, 1.00000, 0.15113)     {};
\node[vertex, label=right:$\vv_3$] (V3) at (0.06012, 1.00000, 0.32500)     {};
\node[vertex] at (0.18500, 0.08603, 1.00000)     {};
\node[vertex, label=right:$\vv_2$] (V2) at (0.06012, 0.18500, 1.00000)     {};

\coordinate (D) at (0.48002068, 0.48002068, 0.48002068);
\coordinate (P1) at (0.060125  , 0.49215023, 0.74561177);
\coordinate (P11) at (-0.03995864,  0.49504136,  0.80891636);

\draw[->,red, thick] (V1) -- (P11);
\draw[->,red, thick] (V2) -- (V3);
\filldraw [black] (P1) circle (0.3pt) node[right, black] {$\vp$};
\filldraw [black] (D) circle (0.3pt) node[right, black] {$\vx$};



\end{tikzpicture}
    \caption{A \acrshort{dbn}-expohedron for $\ndoc=3$ documents with $\vrho=(0.1, 0.5, 0.9)$, $\gamma=0.5$ and $\kappa=0.7$.
    The red arrows illustrate the \acrshort{gls} procedure for the Carathéodory decomposition of point $\vx$.
    The point $\vx$ is decomposed as a convex sum of vertices $\vv_1,\vv_2,\vv_3$.}
    \label{fig:gls}
\end{figure}
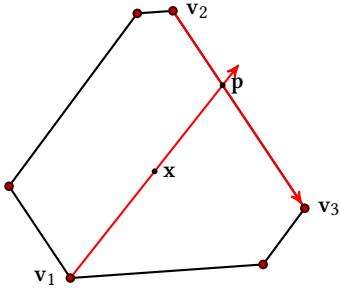

Step (\ref{enum:intersection}) can be solved with a bisection search, provided one has a method for checking whether a point is inside the polytope or not \cite{kletti_introducing_2022}.
For a \acrshort{pbm}, this can be easily checked using the so-called majorization criterion \cite{kletti_introducing_2022,marshall_inequalities_2011}.
For the \acrshort{dbn} we can use \cref{alg:is_inside}.

Step (\ref{enum:vertex_determination}) is much easier.
Given a point $\vv\in\zone(\pi)$ on a face $F$ of the expohedron, the vertex $\vcE(\pi)$ is also on the face $F$, because there is only one vertex in each zone.
Therefore Step \ref{enum:vertex_determination} has complexity $O(\ndoc\log(\ndoc))$.

\Cref{alg:gls} gives the pseudo-code of our Carathéodory decomposition in the \acrshort{dbn}-expohedron.

\begin{algorithm}
	\caption{Our Carathéodory decomposition in the \acrshort{dbn}-expohedron}\label{alg:gls}
	\begin{algorithmic}[1]
		\Procedure{\texttt{GLS}}{\textsc{Input}: $\vx\in\Pi(\gamma,\kappa,\vrho)$}
		    \State\label{l:initial_vertex} $\vv_1\gets\vcE(\argsort(\vx))$ \Comment{Choose the initial vertex}
		    \State $\alpha_1 \gets 1$ \Comment{Set the initial vertex's weight to 1}
		    \State $\vp_1 \gets \vx$
		    \For{$i\in\seq{1}{\ndoc}$}
                \State $\vp_{i+1} \gets \max\{\lambda>0~|~\vx + \lambda(\vx-\vv_i)\in\Pi(\gamma,\kappa,\vrho))\}$ \Comment{Find the intersection with the border of the expohedron using a bisection search with \texttt{is\_inside}}
    			\State $\alpha_{i+1} \gets \alpha_i - \frac{\norm{\vp_i-\vp_{i+1}}}{\norm{\vp_{i+1} - \vv_i}} \alpha_i$ \Comment{Update convex coefficients}
                \State $\alpha_i \gets \frac{\norm{\vp_i-\vp_{i+1}}}{\norm{\vp_{i+1} - \vv_i}} \alpha_i$
                \State $\vv_{i+1} \gets \vcE(\argsort(\vp_{i+1}))]$
    		\EndFor
		\EndProcedure~\textsc{Output}: $\alpha_1,\hdots,\alpha_\ndoc,\vv_1,\hdots,\vv_\ndoc$
	\end{algorithmic}
\end{algorithm}
\begin{theoremEnd}{theo}
    \Cref{alg:gls} finds a Carathéodory decomposition of any point in the \acrshort{dbn}-expohedron with complexity $O(\ndoc^3)$
\end{theoremEnd}
\begin{proofEnd}
    \Cref{alg:gls} implements the \acrshort{gls}, which is proven to yield the Carathéodory decomposition of any point inside the polytope on which it is applied \cite{grotschel_geometric_1993}.
    The most complex operation in the loop has complexity $O(\ndoc^2)$ so the total complexity of the algorithm is $O(\ndoc^3)$.
\end{proofEnd}

\subsection{Pareto-front identification}

Our \acrshort{moo} problem \eqref{eq:moo} can be expressed in the expohedron as
\begin{equation}\label{eq:mooE}
    \max_{\vcE\in\Pi} \vrho^\top\vcE ,\quad \min_{\vcE\in\Pi} \norm{\vcE - \vcEt}_2.
\end{equation}

One objective is linear, the other objective is quadratic.
We can apply a similar procedure as in \cite{kletti_introducing_2022} to geometrically build the Pareto-frontier.
We start from one extremity, the point $\vcEt$, the only one for which the unfairness is minimal.
The level curves of the unfairness are hyperspheres centered in $\vcEt$, so this objective deteriorates in the same manner in every direction.
Because of this, the direction in which we move from here is the one where we increase utility the most within the constraint of staying inside the expohedron.
This direction is the projection of the vector $\vrho$ onto the hyperplane containing the expohedron.
This direction is followed until the border of the expohedron is reached: we are on a face of the expohedron.
Then the same reasoning can be repeated on the current face of the expohedron and we need to follow the direction of the projection of $\vrho$ onto the hyperplane containing the current face.
Each face has at least one less dimension than the previous one, so we are sure to end up at the \acrshort{prp} vertex or at a point on a face with maximal utility

\begin{figure}
    \centering

\begin{tikzpicture}%
	[x={(-0.697158cm, -0.389674cm)},
	y={(0.716917cm, -0.378888cm)},
	z={(-0.000040cm, 0.839404cm)},
	scale=3.000000,
	back/.style={loosely dotted, thin},
	edge/.style={color=black, thick},
	facet/.style={fill=white,fill opacity=0.300000},
	vertex/.style={inner sep=1pt,circle,draw=red!25!black,fill=red!75!black,thick},
	extended line/.style={shorten >=-#1,shorten <=-#1},
    extended line/.default=2cm]
%
%

\coordinate (1.00000, 0.46500, 0.15113) at (1.00000, 0.46500, 0.15113);
\coordinate (1.00000, 0.08603, 0.46500) at (1.00000, 0.08603, 0.46500);
\coordinate (0.32500, 1.00000, 0.15113) at (0.32500, 1.00000, 0.15113);
\coordinate (0.06012, 1.00000, 0.32500) at (0.06012, 1.00000, 0.32500);
\coordinate (0.18500, 0.08603, 1.00000) at (0.18500, 0.08603, 1.00000);
\coordinate (0.06012, 0.18500, 1.00000) at (0.06012, 0.18500, 1.00000);

\coordinate (O) at (2, 2, 2);
\coordinate (D) at (0.48002068, 0.48002068, 0.48002068);
\coordinate (P1) at (0.060125  , 0.49215023, 0.74561177);
\coordinate (P11) at (-0.03995864,  0.49504136,  0.80891636);
\coordinate (V0) at (0.48002068, 0.48002068, 0.48002068);
\coordinate (V1) at (0.060125  , 0.42901383, 0.79790266);
\coordinate (V2) at (0.06012, 0.18500, 1.00000);
\coordinate (V2-rhosur2) at (0.17401273, 0.19883452, 0.9137813);
\coordinate (V2-07rho) at (0.36674582, 0.22224678, 0.76787274);
\coordinate (V2-15rho) at (0.71716962, 0.26481453, 0.50258444);

\coordinate (rhoarrow) at (0.160125, 0.685   , 1.9     );
\coordinate (rholong) at (-0.39603881,  0.37360131,  1.14324142);
\coordinate (X1) at (1,-1,0);
\coordinate (Y1) at (-1/1.732,-1/1.732,2/1.732);

\fill[facet] (0.06012, 0.18500, 1.00000) -- (0.06012, 1.00000, 0.32500) -- (0.32500, 1.00000, 0.15113) -- (1.00000, 0.46500, 0.15113) -- (1.00000, 0.08603, 0.46500) -- (0.18500, 0.08603, 1.00000) -- cycle {};
\draw[edge] (1.00000, 0.46500, 0.15113) -- (1.00000, 0.08603, 0.46500);
\draw[edge] (1.00000, 0.46500, 0.15113) -- (0.32500, 1.00000, 0.15113);
\draw[edge] (1.00000, 0.08603, 0.46500) -- (0.18500, 0.08603, 1.00000);
\draw[edge] (0.32500, 1.00000, 0.15113) -- (0.06012, 1.00000, 0.32500);
\draw[edge] (0.06012, 1.00000, 0.32500) -- (0.06012, 0.18500, 1.00000);
\draw[edge] (0.18500, 0.08603, 1.00000) -- (0.06012, 0.18500, 1.00000);
\node[vertex] at (1.00000, 0.46500, 0.15113)     {};
\node[vertex] at (1.00000, 0.08603, 0.46500)     {};
\node[vertex] at (0.32500, 1.00000, 0.15113)     {};
\node[vertex] at (0.06012, 1.00000, 0.32500)     {};
\node[vertex] at (0.18500, 0.08603, 1.00000)     {};
\node[vertex] (PRP) at (0.06012, 0.18500, 1.00000)     {};




\draw[->,blue] (PRP) -- (rhoarrow);
\draw[thick, red] (D) [x={(X1)},y={(Y1)}] circle (0.198);
\draw[thick, red] (D) [x={(X1)},y={(Y1)}] circle (0.367);
\draw[thick, red] (D) [x={(X1)},y={(Y1)}] circle (0.505);

\draw [extended line,blue] ($(D)!(V2-rhosur2)!(rholong)$) -- (V2-rhosur2);
\draw [extended line,blue] ($(D)!(V2-07rho)!(rholong)$) -- (V2-07rho);
\draw [extended line,blue] ($(D)!(V2-15rho)!(rholong)$) -- (V2-15rho);
\draw [extended line,blue] ($(D)!(V2)!(rholong)$) -- (V2);

\draw[-,green,line width=0.5mm] (V0) --  (V1) node[right,black] {$\vv_1$} -- (V2) node[right, black] {$\vv_2$};
\filldraw[black] (V0) circle (0.3pt) node[left, black] {$\vv_0$};
\filldraw[black] (V1) circle (0.3pt);

\end{tikzpicture}
    \caption{Illustration of \cref{alg:pareto} on a \acrshort{dbn}-expohedron with $\vrho=(0.1,0.5,0.9)$, $\gamma=0.5$ and $\kappa=0.7$. The red circles are the level curves of the fairness objective, the blue lines are the level curves of the utility objective. The blue arrow is the the direction of steepest ascent for the utility objective: the vector $\vrho$ projected on the expohedron's hyperplane. The green curve is the set of Pareto-optimal exposure vectors for demographic fairness. The point $\vv_0$ is the point of pure fairness, where all exposures are identical and the point $\vv_2$ corresponds to a \acrshort{prp} ranking}
    \label{fig:pareto_expohedron}
\end{figure}
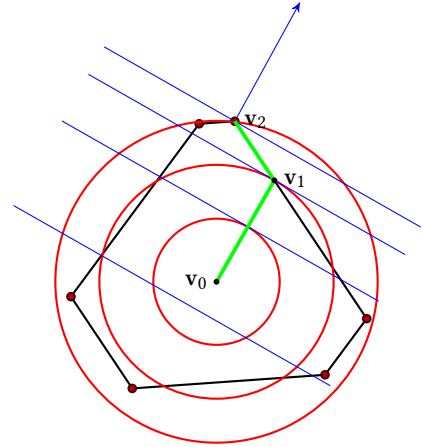

\begin{algorithm}
	\caption{An algorithm able to build the Pareto from of \acrshort{moo} problem \eqref{eq:mooE} in the \acrshort{dbn}-expohedron}\label{alg:pareto}
	\begin{algorithmic}[1]
		\Procedure{\texttt{Pareto}}{\textsc{Input}: Expohedron $\Pi(\gamma,\kappa,\vrho)$, target exposure $\vv^{(0)} = \vcEt$}
			\State $F=(\pi,\splits) \gets \texttt{face\_id}(\vv^{(0)})$ \Comment{Initialize splits}
			\State $P \gets \texttt{Gram-Schmidt}(\{\vnu_s~|~s\in\splits\})$ \Comment{Build the projection matrix $P$ on the linear subspace of $F$}
			\State $\vrho^{(0)} \gets \vrho - P\vrho$ \Comment{Project $\vrho$ on the subspace}
            \State $l \gets 0$
			\While{$\vrho^\top\vv^{(l)} < \vrho^\top\vcE(\piprp)$} \Comment{While utility is not maximal}\label{l:convergence}
    			\State $k \gets \max\{k\geq0~|~\vv^{(l)} + k \vrho^{(l)}\in\Pi\}$ \Comment{Do a bisection search with \texttt{is\_inside}}
    			\State $\vv^{(l+1)} \gets \vv^{(l)} + k \vrho^{(l)}$ 
				\State $F=(\pi,\splits) \gets \texttt{face\_id}(\vv^{(l+1)})$ 
				\State $P \gets \texttt{Gram-Schmidt}(\{\vnu_s~|~s\in\splits\})$ \Comment{Build the projection matrix on the linear subspace of the new face $F$}
				\State $\vrho^{(l+1)} \gets \vrho - P\vrho$  \Comment{Project $\vrho$ on the new face}
				\State $l\gets l + 1$
			\EndWhile
		\EndProcedure\\\textsc{Output}: A sequence of at most $\ndoc$ points $(\vv^{(j)})_{j\in\seq{0}{l}}$.
		The Pareto-curve is the union of line segments of endpoints $[\vv^{(j)},\vv^{(j+1)}]$.
	\end{algorithmic}
\end{algorithm}
\begin{theoremEnd}{theo}\label{theo:pareto}
    \Cref{alg:pareto} finds all Pareto-optimal exposure vectors in the expohedron $\Pi(\gamma,\kappa,\vrho)$.
\end{theoremEnd}
\begin{proofEnd}
    The fact that the algorithm builds the Pareto front is sufficiently explained in the paragraph preceding the pseudo-code.
    The complexity of the most complex operation inside the loop is $O(\ndoc^2)$ and the loop has length at most $\ndoc$, so the total complexity of the algorithm is $O(\ndoc^3)$.
\end{proofEnd}

\section{Experiments}
    We perform experiments using two datasets: TREC and MSLR.
Our code will be made available online.\footnote{\url{https://github.com/naver/expohedron}}
Our experiments were performed on a laptop with an Intel\textregistered Core\texttrademark i7-8650U CPU @ 1.90GHz processor.

\subsection{Datasets}
We use the TREC2020 Fairness Track dataset\footnote{\url{https://fair-trec.github.io/2020/index.html}} and the \acrshort{mslr} dataset \cite{DBLP:journals/corr/QinL13}.
The \acrshort{trec} data is actually made up of two datasets: the training set and the test set.
Since our own method does not require any training, we evaluate it on the test set only, but we use the training set on \ltr~, one of our baselines (see below).
The training set contains $200$ queries with a variable number of documents ranging from $10$ to $171$.
For each query-document pair, we estimated relevances scores with values in $[0,1]$ ourselves.
The details of how we computed them will be made available in our github repository.

For the \acrshort{mslr} dataset, we randomly separated it into a training set and a test set of $434$ queries each.
We use the test set for evaluation and we use the training set only for the baseline method that requires one.
For each query-document pair, a graded relevance value in $\{0,1,2,3,4\}$ is provided by the dataset.
We divide those values by $4$ in order to conform to our setting that requires relevance values to be in $[0,1]$.
Thus for \acrshort{mslr}, the relevance vector $\vrho$ is in $\{0, 0.25, 0.5, 0.75, 1\}^\ndoc$.

\subsection{Metrics and baselines}
We use a \acrshort{dbn} exposure model as described by \eqref{eq:dbn}, with parameters $\gamma=0.5$ and $\kappa=0.7$.
These parameters are purely arbitrary, and they correspond to the parameters used in the TREC2020 challenge.
We consider the merit of a document to be its relevance $\vmu=\vrho$, \acrshort{ie} we work with meritocratic fairness.
Therefore the exposure vector guaranteeing fairness is the only $\vcEt\in\Pi(\gamma,\kappa,\vrho)$ such that $\vcEt\propto\vrho$.
If it does not exists we take $\vmu'=\mu + k\vone$ with $k>0$ as small as possible.
We compute this value using a bisection search with \cref{alg:is_inside}.


We analyze 4 methods in total.
For each method, $\horizon=1\,000$ rankings are delivered.
Exposures cumulate over time and we average the final exposures by the time horizon $\horizon$, \acrshort{ie} the total number of rankings delivered for each query, so as to get a metric that does not depend on $\horizon$.
For the methods requiring a scalarized objective we use
\begin{equation}
    \min \alpha (-\nutility) + (1-\alpha)\nunfairness^2.
\end{equation}
We make $\alpha$ vary in $[0,1]$ with a step $0.05$, getting $21$ values of $\alpha$.
We use the square $\nunfairness^2$ instead of $\nunfairness$, because this simplifies the gradient computation for the baseline \ltr.

The methods we compare are the following:
\begin{enumerate}
    \item Our own method (\expo) builds the Pareto frontier in the expohedron. It then chooses a point on this Pareto frontier corresponding to a trade-off $\alpha$ between $\nutility$ and $\nunfairness$.
    It expresses it as the expected value of a distribution over permutations.
    Finally it delivers the sequence of $\horizon$ rankings using $m$-balanced words \cite{sano_m-balanced_2004,kletti_introducing_2022}.
    Intuitively, $m$-balanced words are sequences of letters such that in every sub-string the frequencies of each letter are as equal as possible, given an overall proportion in the whole string.
    In our case using balanced words instead of randomly sampling from our distribution $\cD$ allows to more accurately approach the expected value of $\cD$ at any time.
    \item We execute a controller (\ctrl) \cite{thonet_multi-grouping_2020} and recover the fairness and utility at the last time step.
    This controller computes at each time step the current error with respect to the target exposure and defines the scores at time $t$ as
    \begin{equation}
        \vy_t = \vrho + g(\vcEt - \vcE_t),
    \end{equation}
    where $\vcE_t$ is the exposure vector at time $t$ and $g\geq0$ is a trade-off parameter called \emph{gain}.
    At time $t+1$ the ranking is determined by decreasing scores.
    We make the gain vary in $[0,1]$ on a logarithmic scale.
    \item We sample from a \acrlong{pl} distribution (\pl) as proposed by \citeauthor{diaz_evaluating_2020} \cite{diaz_evaluating_2020} with varying temperature $\temperature$ and with parameters $\vrho$.
    We make the temperature vary in $[0.001, 50]$ on a logarithmic scale.
    \item We use the approach by \citeauthor{oosterhuis_computationally_2021} \cite{oosterhuis_computationally_2021} (\ltr), that we adapted to the \acrshort{dbn} exposure model.
    This baseline requires a training set to find an optimal function $\sigma:\R\to\R$ that to a relevance value $\vrho_i$ associates a score $\sigma(\vrho_i)$ to be used as log-scores of a \acrshort{pl} distribution.
    The function $\sigma$ is then used on the test set and the rankings are sampled by a policy $\PL(\sigma(\vrho_1),\hdots,\sigma(\vrho_\ndoc))$.
\end{enumerate}

\subsection{Results}

\paragraph{Does our method Pareto-dominate the baselines?}
From \cref{theo:pareto} we expect the method \expo~ to yield Pareto-optimal expected exposure vectors.
We experimentally verify this claim by comparing its performance to the baselines in terms of the two objectives $\nutility,\nunfairness$.

\Cref{fig:pareto} confirms that our method Pareto-dominates all the baselines.
We observe that the \ctrl~ baseline is particularly close to the \expo~ method.
This suggests that the \ctrl~ baseline might be Pareto-optimal, in some cases; and it is an interesting empirical result in itself since this method was proposed as a heuristic and has as of yet no proof of optimality.
However \cref{fig:pareto_mslr} also shows that for some parameters, \ctrl~ yields solutions that are not Pareto-optimal (point $(1,1)$).
Since we work with meritocratic fairness and, since with a very high temperature a \acrshort{pl} policy delivers rankings in a uniform manner, it comes as no surprise that the curve obtained with the \pl~ method does not attain zero unfairness and that the unfairness does not attain its minimum for an infinite temperature.
The \ltr~ method performs expectedly better than a naive \acrshort{pl} policy, but it is not as close to the true Pareto-front as the \ctrl~ method, especially at the pure fairness endpoint, which it does not reach. 

The result for the \acrshort{mslr} dataset in \cref{fig:pareto_mslr} present a particularity in that the \pl, \ltr~ and \expo~ curves have a normalized unfairness lower than 1, while being of maximal utility, whereas \ctrl~ reaches the point $(1,1)$ for a gain of $0$.
In order to achieve $\nunfairness<1$ and $\nutility=1$, the \ctrl~ method would seem to require a gain slightly larger than $0$.
This can be explained by the fact that the \acrshort{mslr} dataset contains a lot of duplicates relevance values, as the values are all in $\{0, 0.25, 0.5, 0.75, 1\}$.
This means that there exist many different \acrshort{prp} rankings, and delivering several different \acrshort{prp} rankings can already reduce the unfairness by a large amount without reducing the utility.
This happens naturally for a \acrshort{pl} policy, because items with equal scores will have equal expected exposure, it happens for \expo~ because it is Pareto-optimal, but it does not happen for the controller when the gain is set to $0$.

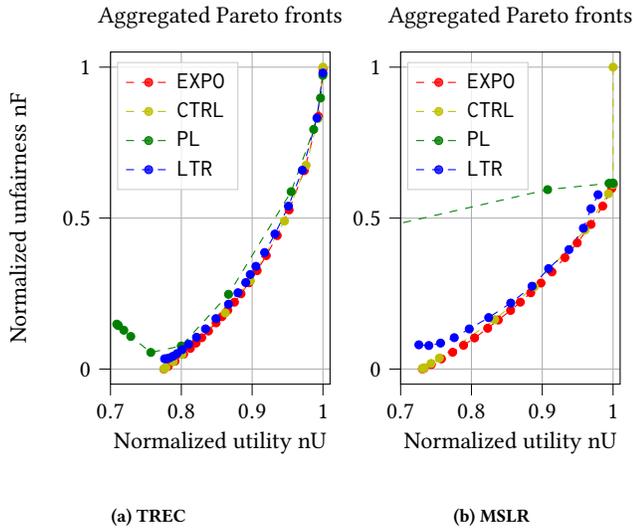
\begin{figure}
    \centering
    \begin{subfigure}[b]{0.22\textwidth}
        \centering
\begin{tikzpicture}

\definecolor{color0}{rgb}{0.75,0.75,0}

\begin{axis}[
legend cell align={left},
legend style={
  fill opacity=0.8,
  draw opacity=1,
  text opacity=1,
  at={(0.03,0.97)},
  anchor=north west,
  draw=white!80!black
},
tick align=outside,
tick pos=left,
title={Aggregated Pareto fronts},
x grid style={white!69.0196078431373!black},
xlabel={Normalized utility $\nutility$},
xmajorgrids,
xmin=0.7, xmax=1.01,
xtick style={color=black},
y grid style={white!69.0196078431373!black},
ylabel={Normalized unfairness $\nunfairness$},
ymajorgrids,
ymin=-0.05, ymax=1.05,
ytick style={color=black},
width=4.5cm,
height=6cm
]
\addplot [red, dashed, mark=*, mark size=1.5, mark options={solid}]
table {%
0.775375112121316 0
0.78147977572292 0.0092678044453467
0.79090695009833 0.0259683778173305
0.803147648979073 0.0491738271428365
0.812760337345936 0.0685431623164238
0.820764470348525 0.0859938188482096
0.828807151983882 0.103741018904232
0.838351243556377 0.125994503408065
0.849111517505538 0.15298432689513
0.85742331442521 0.173517614502826
0.865568222913163 0.19478067388471
0.874984554951533 0.221485084992183
0.8842606854668 0.249416186374157
0.895500392797671 0.285828046213524
0.907068511058828 0.326000608551475
0.919971844111792 0.375919466314644
0.935379026716027 0.442065829704896
0.952189142837411 0.527073244888376
0.973603131811364 0.657297144276602
0.993643987525748 0.836743953793892
0.999999999944188 0.993768137408836
};
\addlegendentry{\expo}
\addplot [color0, dashed, mark=*, mark size=1.5, mark options={solid}]
table {%
1 1
0.993146224778198 0.829667283077413
0.976182275754992 0.674972276086076
0.945657737091824 0.490224198182762
0.897559815985678 0.29242012933644
0.86206614140096 0.185943465147129
0.800883252235699 0.0480505961501712
0.788564982179503 0.0244803505007316
0.778281833815131 0.00537154282789199
0.776923054900183 0.00302520767840848
0.775845383319564 0.00148870056769898
0.775705817801732 0.00137364435481179
};
\addlegendentry{\ctrl}
\addplot [green!50!black, dashed, mark=*, mark size=1.5, mark options={solid}]
table {%
0.999847307483533 0.972703861808963
0.996588224627014 0.897865731519242
0.986894450469216 0.793706010807349
0.955078520410327 0.587681201318961
0.866658850790854 0.247589843876544
0.79992424806472 0.0763737736725384
0.756905070160893 0.0550759661564176
0.728556793751648 0.108018329372883
0.718803435320546 0.128579940805523
0.711114686910817 0.14390096672583
0.709866111114487 0.146526830480068
0.708789676522604 0.148818482849687
};
\addlegendentry{\pl}
\addplot [blue, dashed, mark=*, mark size=1.5, mark options={solid}]
table {%
0.776375462740693 0.0340230288900443
0.779645524263751 0.0345373380443067
0.783579051380052 0.0371398272422996
0.788226576294388 0.0424735726957682
0.793894591202173 0.0511939752292915
0.801074704169771 0.0641563932089553
0.810347152162433 0.0824786140141494
0.82140462427175 0.105410790035226
0.83413817731938 0.133007844435218
0.848868609964267 0.167122041371429
0.86663747885326 0.214769946396492
0.879972587462328 0.252772264957569
0.890950909925444 0.286493500596741
0.89734220309137 0.31310673258721
0.905347299116849 0.340359309183684
0.917647778328338 0.385950342506814
0.932196878946112 0.447293863559039
0.951256342048539 0.539520174918245
0.970872763563049 0.65862105964256
0.991507775827483 0.831539136428047
0.999883868464372 0.980672012009467
};
\addlegendentry{\ltr}
\end{axis}

\end{tikzpicture}
        \caption{\acrshort{trec}}
    \label{fig:pareto_trec}
    \end{subfigure}
    \hfill
    \begin{subfigure}[b]{0.22\textwidth}
        \centering
\begin{tikzpicture}

\definecolor{color0}{rgb}{0.75,0.75,0}

\begin{axis}[
legend cell align={left},
legend style={
  fill opacity=0.8,
  draw opacity=1,
  text opacity=1,
  at={(0.03,0.97)},
  anchor=north west,
  draw=white!80!black
},
tick align=outside,
tick pos=left,
title={Aggregated Pareto fronts},
x grid style={white!69.0196078431373!black},
xlabel={Normalized utility $\nutility$},
xmajorgrids,
xmin=0.7, xmax=1.01,
xtick style={color=black},
y grid style={white!69.0196078431373!black},
ymajorgrids,
ymin=-0.05, ymax=1.05,
ytick style={color=black},
width=4.5cm,
height=6cm
]
\addplot [red, dashed, mark=*, mark size=1.5, mark options={solid}]
table {%
0.730701926960178 0
0.733385855091283 0.00275745789491244
0.743372078655511 0.0151117907771046
0.757571206781209 0.0337358697619731
0.773466542520624 0.055641830771297
0.788887414234442 0.0786817715044277
0.804451276872889 0.103015498734132
0.822905374040111 0.135174482333407
0.837879858342146 0.16251133580729
0.855270313421638 0.194222440581097
0.869021263566367 0.221605361120982
0.883804161573403 0.252639719944466
0.898212403095502 0.284783975101793
0.913681050641211 0.321635953291831
0.932048741025794 0.369141285553034
0.949394896048568 0.417934486176822
0.968761319149453 0.47900241336683
0.985415161637599 0.539811658491438
0.998176114920149 0.598331254111569
0.999974933025333 0.6102683298091
1.0000000002684 0.610821100453368
};
\addlegendentry{\expo}
\addplot [color0, dashed, mark=*, mark size=1.5, mark options={solid}]
table {%
1 1
1 0.610967708678216
0.993495565871043 0.580861367201155
0.960395871752703 0.460866799869364
0.888786042124612 0.272631936570754
0.832866610270991 0.163725881774665
0.754733218355683 0.036478370267319
0.742890687856838 0.0185210184821848
0.733296931690685 0.00421198542692066
0.732094591254442 0.00256014634780036
0.730937454592573 0.00168211299870735
0.73076258694647 0.00168022773897939
};
\addlegendentry{\ctrl}
\addplot [green!50!black, dashed, mark=*, mark size=1.5, mark options={solid}]
table {%
1 0.615946042623925
0.99998651947875 0.615108364628752
0.994294009200546 0.615260500884104
0.907611282014481 0.594314911418484
0.595186584876591 0.426722563050589
0.39782028673949 0.464934026166987
0.307801903403197 0.540573340410815
0.262211316981787 0.582700384612825
0.247919856992604 0.595240725128647
0.237283098400586 0.60534889698625
0.236365731284403 0.605926698490384
0.234854136767845 0.607239686530626
};
\addlegendentry{\pl}
\addplot [blue, dashed, mark=*, mark size=1.5, mark options={solid}]
table {%
0.725727811500187 0.0802891966886857
0.739859125001145 0.0779161966160102
0.756408706236698 0.0860151928543927
0.775742428605901 0.103975369689734
0.796919887673006 0.132884477537166
0.824341079076 0.170486714781877
0.855489673387074 0.218661737653618
0.885576116883649 0.274311683661772
0.908991211670441 0.33285059861954
0.937722023746231 0.395710724642321
0.958109871679382 0.466492133538069
0.968762403777808 0.531071912953751
0.97878462337593 0.577500328638093
};
\addlegendentry{\ltr}
\end{axis}

\end{tikzpicture}
        \caption{\acrshort{mslr}}
    \label{fig:pareto_mslr}
    \end{subfigure}
    \caption{For both the \acrshort{trec} and the \acrshort{mslr} dataset, and for our baselines the averaged normalized unfairness metric and the average normalized utility metric are plotted for various trade-off parameters after delivery of $\horizon=1\,000$ rankings.}
    \label{fig:pareto}
\end{figure}

\paragraph{How does the fairness evolve over time?}
Assume that we are in a situation in which our goal is to minimize the unfairness only, without consideration for the utility.
We choose for each method the parameters $\alpha$, $\tau$ or $g$ that minimize the normalized unfairness.
Then we look at how the average normalized unfairness metric evolves over the delivery of $\horizon=1\,000$ rankings, for our different methods.

\begin{figure}
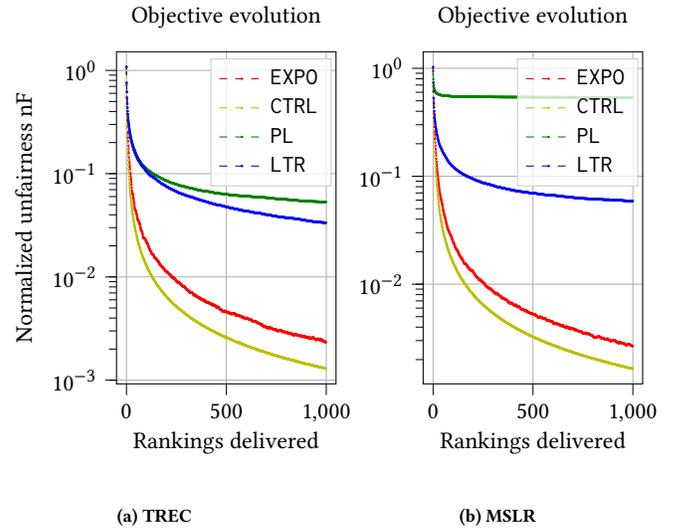

    \centering
    \begin{subfigure}[b]{0.22\textwidth}
        \centering
        \input{fig/convergence_plots_TRECmeritocratic1000}
        \caption{\acrshort{trec}}
    \label{fig:convergence_trec}
    \end{subfigure}
    \hfill
    \begin{subfigure}[b]{0.22\textwidth}
        \centering
        \input{fig/convergence_plots_MSLRmeritocratic1000}
        \caption{\acrshort{mslr}}
    \label{fig:convergence_mslr}
    \end{subfigure}
    \caption{For both the \acrshort{trec} and the \acrshort{mslr} dataset, and for our baselines both the unfairness metric and the utility metric are plotted for various parameters after delivery of $\horizon=1\,000$ rankings.}
    \label{fig:convergence}
\end{figure}

The results, displayed in \cref{fig:convergence}, show that with the \pl~method, the metric does not converge to $0$, as can be expected from \cref{fig:pareto}.
Furthermore the \ctrl~method seems to converge more quickly to $0$ than the \expo~method.
This may be due to the fact that the \ctrl~method is at liberty of choosing from any possible ranking at each time step, while the \expo~method, delivered with balanced words, has a support of size at most $\ndoc$.
Recall, however, that convergence to $0$ of the controller (let alone "faster than \expo") is merely an experimental observation, whereas it is theoretically guaranteeed with \expo.

\paragraph{Is our method efficient?}
The time taken by the methods cannot be fairly compared by one number alone.
For instance our \expo~ method has a moderate initial computation cost for each query, but once an optimal distribution $\cD$ is found, a policy can be delivered very efficiently.
On the other hand the \ctrl~ method has no initial cost, but some computations need to be done each time a new ranking is delivered.
The \ltr~ method has a very high initial cost of training, but once a mapping is found the cost of ranking delivery is low.
The \pl~ method has no initial cost and it can be delivered efficiently.
To account for these subtleties, we report separately for each method the time for pre-computation (which is constant) and the time for delivery (which is proportional to the number $\horizon$ of rankings delivered).

In \cref{tab:runtimes} we report the average time it takes to deliver $T=1\,000$ rankings for each query, for both datasets.
The results show the extraordinary efficiency of using an $m$-balanced word generator to deliver a ranking policy found with the \expo~ method.
They also show the efficiency of our Carathéodory decomposition algorithm and of our Pareto-front identification algorithm, especially \acrshort{wrt} a learning approach \ltr.
Finally we note that the \ctrl~ method has no initial cost, but a higher delivery cost per ranking, than our \expo~ method.
Hence it may be faster for short time horizons $\horizon$, but will be slower for longer horizons.

\begin{table}[h]
    \centering
    \caption{For both datasets and a given trade-off parameter, we report the time in seconds, that it took to deliver $\horizon=1\,000$ rankings.
    \expo~Pareto is the time it took to compute the whole Pareto-front using our \expo~ method, averaged over all the queries.
    \expo~ \acrshort{gls} is the time it took to compute the Carathéodory decomposition of the fairness endpoints, averaged over all the queries.
    \expo~ delivery is the time it took to deliver a total of $\horizon=1\,000$ rankings using an $m$-balanced word generator, averaged over all the queries.
    \ltr~ learning is the total training time of the whole training dataset, while \ltr~delivery it the average time it took to deliver a total of $\horizon=1\,000$ rankings by sampling a \acrshort{pl} distribution.
    \ctrl~ is the time it took to deliver a total of $\horizon=1\,000$ rankings using the \ctrl~ method, averaged over all the queries.}
    \tabcolsep=0.11cm
    \begin{tabular}{|c|c|c|}
        \hline
        & \acrshort{trec} & \acrshort{mslr}\\
        \hline
        \expo~ Pareto & 0.9962 & 0.1853 \\
        \expo~ \acrshort{gls} & 1.0591 & 4.350 \\
        \expo~ delivery & 0.0019 & 0.0054 \\
        \ctrl & 0.0568 & 0.0796 \\
        \pl & 0.0141 & 0.0630 \\
        \ltr~learning & 468.5 & 1746.2 \\
        \ltr~delivery & 0.0143 & 0.0631 \\
        \hline
    \end{tabular}
    \label{tab:runtimes}
\end{table}

\section{Conclusion}
    
In this paper, we propose a geometric method to compute Pareto-optimal ranking policies to maximize both user utility and item producer fairness.
Our method is the first that provably solves this key timely problem for a realistic \acrshort{dbn} exposure model---while prior literature was limited to a \acrshort{pbm} exposure model---; and it works in reasonable time and space complexity.
Our experiments show that it outperforms classical models based on \acrshort{pl} distributions, which are not optimal. Interestingly though, we find that a simple heuristic controller does (almost) as good as our method in terms of optimality.
This is interesting as it experimentally shows optimality of the controller in our experiments, and also because it offers an alternative to adapt to the situation at stake: the controller has no pre-computation time and is faster for delivering just a few rankings; whereas our method has a (small) offline pre-computation time, but is much faster per delivered ranking.
Of course our method also has the advantage of the theoretical guarantee of optimality in any case.

\begin{acks}
    This work has been partially supported by MIAI@Grenoble Alpes, (ANR-19-P3IA-0003).
\end{acks}

\clearpage
\newpage
\bibliographystyle{ACM-Reference-Format}
\balance
\bibliography{SIGIR2022,biblio}
\clearpage
\vspace{0.1\textheight}

\appendix
\section*{Appendices}
\renewcommand{\thesubsection}{\Alph{subsection}}
    \subsection{Proofs}
        \printProofs
    \subsection{Click Models}\label{app:click}
        We will refer tho the click models described in the book \cite{chuklin_click_2015} and, in particular, use the same notations. 

\subsubsection{\acrfull{dbn}}
The general recursive formulation of the \acrshort{dbn} Exposure (\acrshort{aka} Examination probability) for the query-document pair $(q,d)$ at rank $k$ is:
\begin{equation}
    \epsilon_k = \epsilon_{k-1} \gamma (1 - \alpha_{q,d} \sigma_{q,d})
\end{equation}
where $\gamma$ is the patience parameter, $\alpha_{q,d}$ is the attraction probability of document $d$ with respect to $q$ and $\sigma_{q,d}$ is the probability that the user will be satisfied with that document and stop looking further down in the list. In this model, the product $\alpha_{q,d} \sigma_{q,d}$ is the relevance probability i.e. $\P(r_{q,d}=1) = \alpha_{q,d} \sigma_{q,d}$.
This model is equivalent to our general expression in \eqref{eq:dbn}, with $\rho = \alpha_{q,d} \sigma_{q,d}$ and $\kappa=1$.

\subsubsection{\acrfull{sdbn}}
The \acrshort{sdbn} Model is the particular case of the \acrshort{dbn} model with a patience parameter $\gamma$ equal to one. Its general recursive formulation for the query-document pair $(q,d)$ at rank $k$ is:
\begin{equation}
    \epsilon_k = \epsilon_{k-1} (1 - \alpha_{q,d} \sigma_{q,d})
\end{equation}
Define $\omega:=\frac{1}{2}\min\alpha_{q,d} \sigma_{q,d}$.
Assuming that $\alpha_{q,d} \sigma_{q,d}$ is always in $(0,1)$, this model is equivalent to our general expression in \eqref{eq:dbn}, with $\gamma = 1 - \omega$, and $\kappa = 1$ and $\rho=1-\frac{1-\alpha_{q,d} \sigma_{q,d}}{1-\omega}$.

\subsubsection{\acrfull{cm}}
The Cascade Model is the particular case of the \acrshort{dbn} model with a patience parameter $\gamma$ equal to one and a satisfaction probability $\sigma_{q,d}$ also equal to one. In other words, it assumes that if a document snippet looks relevant, then the document is also relevant and the user will be fully satisfied after clicking on it; consequently, she will stop after one single click. Its general recursive formulation for the query-document pair $(q,d)$ at rank $k$ is:
\begin{equation}
    \epsilon_k = \epsilon_{k-1} (1 - \alpha_{q,d})
\end{equation}
In this model, the term $\alpha_{q,d}$ is considered as the relevance probability i.e. $\P(r_{q,d}=1) = \alpha_{q,d}$.
This model is equivalent to the \acrshort{sdbn}, with $\sigma_{q,d}=1$ and therefore also equivalent to our general expression \eqref{eq:dbn} if all $\alpha_{q,d}\in(0,1)$.

\subsubsection{\acrfull{dcm}}
The \acrshort{dcm} generalizes the \acrfull{cm} by still allowing a user to go down in the list even after clicking on one document. In other words, it allows the user to click on more than one document. The general recursive formulation of the \acrshort{dcm} Exposure for the query-document pair $(q,d)$ at rank $k$ is:
\begin{equation}
    \epsilon_k = \epsilon_{k-1} (1 - \alpha_{q,d} (1 - \lambda))
\end{equation}
where $\lambda$ is the probability of continuing the examination of the next items in the list after clicking on a document.
This model is equivalent to the \acrshort{sdbn} with $\sigma_{q,d}=(1 - \lambda)$ and therefore also equivalent to our general expression \eqref{eq:dbn} if all $\alpha_{q,d}(1-\lambda)\in(0,1)$.

Note that some versions of the \acrshort{dcm} exist with a rank-dependent $\lambda_k$, instead of a single continuation parameter $\lambda$.
These versions do not fit our general expression. 

\subsubsection{\acrfull{ccm}}
The \acrshort{ccm} also generalizes the \acrfull{cm}, but in a much more flexible way. The general recursive formulation of the \acrshort{ccm} Exposure for the query-document pair $(q,d)$ at rank $k$ is:
\begin{equation}
    \epsilon_k  =  \epsilon_{k-1} (\alpha_{q,d} ((1-\alpha_{q,d} )\tau_2 +\alpha_{q,d} \tau_3) + (1 - \alpha_{q,d})\tau_1) 
\end{equation}
which can be re-written as:
\begin{equation}
    \epsilon_k  =  \epsilon_{k-1} \tau_1 (1- \kappa_{q,d})
\end{equation}
with $\kappa_{q,d} = \alpha_{q,d} \left( \frac{\tau_1-\tau_2}{\tau_1} + \alpha_{q,d} \frac{\tau_2-\tau_3}{\tau_1}\right)$. See \cite{chuklin_click_2015} for the meaning of the $\tau_i$ parameters. 

This model is equivalent to our general expression in \eqref{eq:dbn}, with $\gamma = \tau_1$ and $\rho = \alpha_{q,d} \left( \frac{\tau_1-\tau_2}{\tau_1} + \alpha_{q,d} \frac{\tau_2-\tau_3}{\tau_1}\right)$ and $\kappa=1$.

\end{document}